\documentclass{article}
\usepackage{graphicx} 
\usepackage{algorithm}
\usepackage{algorithmic}
\usepackage[utf8]{inputenc}
\usepackage{natbib}
\usepackage{amsmath,amssymb}
\usepackage{bbm}
\DeclareMathOperator*{\argmax}{arg\,max}

\usepackage{hyperref}
\usepackage{cleveref}

\usepackage{amsthm}
\newtheorem{lemma}{Lemma}

\newtheorem{definition}{Definition}

\newtheorem{theorem}{Theorem}

\usepackage{fullpage}

\usepackage{graphicx}

\usepackage{comment}

\usepackage[shortlabels]{enumitem}

\newcommand{\E}{\mathbb{E}}

\usepackage[shortlabels]{enumitem}

\newcommand{\ignore}[1]{}

\newcommand{\hide}[1]{}

\usepackage[shortlabels]{enumitem}

\newcommand{\coloneqq}{:=}

\usepackage[shortlabels]{enumitem}

\newcommand{\A}{\mathcal{A}}

\usepackage{caption}
\usepackage{xcolor}
\include{bib}
\usepackage{amsmath}
\usepackage{xcolor}
\usepackage{natbib}
\usepackage{mathtools}
\usepackage{authblk}

\newtheorem{assumption}{Assumption}
\newtheorem{remark}{Remark}

\newcommand{\cA}{\mathcal{A}}
\newcommand{\cD}{\mathcal{D}}
\newcommand{\cP}{\mathcal{P}}
\newcommand{\cU}{\mathcal{U}}
\newcommand{\cY}{\mathcal{Y}}
\newcommand{\cR}{\mathcal{R}}
\newcommand{\cT}{\mathcal{T}}

\makeatletter
\newenvironment{game}[1][htb]{%
    \renewcommand{\ALG@name}{Game}
   \begin{algorithm}[#1]%
  }{\end{algorithm}}
\makeatother

\title{Repeated Contracting with Multiple Non-Myopic Agents:\\ Policy Regret and Limited Liability}
\author{Natalie Collina} 
\author{Varun Gupta} 
\author{Aaron Roth}
\affil{Department of Computer and Information Sciences, University of Pennsylvania}
\begin{document}
\begin{titlepage}
\maketitle

\begin{abstract}
We study a repeated contracting setting in which a Principal adaptively chooses amongst $k$ Agents at each of $T$ rounds. The Agents are non-myopic, and so a mechanism for the Principal induces a $T$-round extensive form game amongst the Agents. We give several results aimed at understanding an under-explored aspect of contract theory --- the game induced when choosing an Agent to contract with. First, we show that this game admits a pure-strategy \emph{non-responsive} equilibrium amongst the Agents --- informally an equilibrium in which the Agent's actions depend on the history of realized states of nature, but not on the history of each other's actions, and so avoids the complexities of collusion and threats. Next, we show that if the Principal selects Agents using a \emph{monotone} bandit algorithm, then for any concave contract, in any such equilibrium, the Principal obtains no regret to contracting with the best Agent in hindsight --- not just given their realized actions, but also to the counterfactual world in which they had offered a guaranteed $T$-round contract to the best Agent in hindsight, which would have induced a different sequence of actions. Finally, we show that if the Principal selects Agents using a monotone bandit algorithm which guarantees no swap-regret, then the Principal can additionally offer only limited liability contracts (in which the Agent never needs to pay the Principal) while getting no-regret to the counterfactual world in which she offered a linear contract to the best Agent in hindsight --- despite the fact that linear contracts are not limited liability. We instantiate this theorem by demonstrating the existence of a monotone no swap-regret bandit algorithm, which to our knowledge has not previously appeared in the literature. 

\end{abstract}
\setcounter{page}{0}
 \thispagestyle{empty}
 \clearpage
\tableofcontents
 \thispagestyle{empty} \setcounter{page}{0}
 \clearpage

\end{titlepage}

\section{Introduction}

Principal-Agent problems arise when a Principal seeks to hire a third party (an Agent) to perform some costly task on their behalf. Typically the action of the Agent (e.g. the effort that they put in) is unobservable to the Principal --- all that the Principal observes is the resulting outcome, which is a function both of the Agent's action and some unknown \emph{state of nature}. The design problem for the Principal is to find a contract (a mapping from outcomes to payments to the Agent) that maximizes their expected utility less the payment, when the Agent best responds so as to maximize their expected payment, less their cost. Classical contract theory concerns itself with the design of the optimal contract for the Principal, in a static setting with a single Agent, in which the Principal and the Agent share a common prior belief about the distribution over states of nature. This standard setting is well understood, and \emph{linear contracts} (which pay the Agent a constant fraction of the Principal's realized value) turn out to be robustly optimal in several different senses \citep{carroll2015robustness,dutting19,dutting2022combinatorial,castiglioni2021bayesian}. 

Of course, each of the standard assumptions can be subject to challenge. For example, how does the Principal choose the Agent to hire? In many markets there is an aspect of competition between Agents that is unmodelled in a standard Principal/Agent setting. What if the distribution over states of nature deviates from the expectations of the Principal and Agent---i.e. what if their beliefs are mis-specified? In this case, what guarantees can be given to the Principal?  And finally, many contracts (like linear contracts) can require \emph{negative payments} to the Agent in certain realizations --- i.e. they can obligate the Agent to pay the Principal, which is generally unimplementable in practice. Contracts that never obligate the Agent to pay the Principal are said to satisfy the \emph{limited liability} condition, and are more realistic, but complicate the design of optimal contracts.  These are the problems that we focus on in this paper.

As a running example, consider a college endowment that wants to hire a fund manager to invest it's capital. The goal of the endowment is to maximize the return on its investment (minus the fees it pays).  There are $k$ fund managers that the endowment can choose between, each of which must be paid a fee, as well as a costless outside option (say, investing the money in an index fund). A linear contract in this setting pays the Agent a fixed percentage (say 10\%) of its excess performance above that of the index fund. Since this might be negative (if the Agent does not outperform the index fund), a linear contract is not limited liability. 

In our running example, the problem that the Principal needs to solve in the most complex setting considered in our paper is to choose a fund manager amongst the $k$ candidates, or the outside option in each of $T$ sequential rounds. In each round, the Principal offers the chosen fund manager a contract, which must be limited liability. In fixing a method for doing this, the Principal is inducing a $T$ round extensive form game amongst the $k$ fund managers, who are not myopic, and who will play in an equilibrium of this game, given their beliefs about the distribution on the sequence of states of nature to be realized.  We wish to guarantee to the Principal that in the equilibrium of this induced game, their cumulative utility is at least what it would have been had they selected the best of the fund managers in hindsight and offered them a (non-limited-liability) linear contract, guaranteed for all $T$ rounds, no matter what the sequence of states of nature turn out to be (i.e. even if the beliefs of the fund managers turn out to be incorrect). Note that what we are aiming for is a form of \emph{policy regret}, in that in the actual play-out, the fund managers are playing an equilibrium strategy in a $k$ player extensive form game, but in the counterfactual that we use as our benchmark, the fund managers are guaranteed a fixed contract over all $T$ rounds, and so are simply best responding given their beliefs.

We now lay out our results in a bit more detail.

We study a repeated Principal/Agent problem over $T$ rounds. We model the action space for the Agent as being the set of real valued \emph{effort levels} $a \in [0,1]$, associated with a convex increasing cost function $c(a)$. In each round $t$:

\begin{enumerate}
    \item The Principal chooses amongst a set of $k$ Agents as well as a costless \emph{outside option}. If the Principal chooses one of the $k$ Agents, he offers them a concave \emph{contract} $v^t:\mathbb{R}\rightarrow \mathbb{R}$ (a mapping from Principal payoffs to Agent payments) for that round. 
    \item The chosen Agent then chooses an action $a^t$. Then, a state of nature $y^t$ is realized according to an arbitrary (possibly adversarial) process, which results in payoff $r(a^t,y^t)$ for the Principal and a payment $v^t(r(a^t,y^t))$ for the Agent. In that round, the Principal gets utility $r(a^t,y^t) - v^t(r(a^t,y^t))$ and the Agent gets utility $v^t(r(a^t,y^t)) - c(a^t)$. We assume throughout that for each state of nature $y$, the Principal's reward $r(a,y)$ is a non-decreasing linear function of the Agent's effort level $a$ (and can depend on the state of nature $y$ in arbitrary ways).
\end{enumerate}
At the end of each round, the Principal as well as the Agents observe $y^t$ and $r(a^t,y^t)$, which can be used to inform the Principal's choice of Agents in the following rounds. counterfactuals of the sort ``what \emph{would} the reward have been had the Principal chosen a different Agent'' are not observed. A strategy for an Agent in this game is a mapping from transcripts (which records both the realized states of nature and actions of the other players at all previous rounds) and current-round-contracts to effort levels. The Agents each have (possibly mis-specified) beliefs over the distribution of the state of nature at each round (which can also be arbitrary functions of the transcript). An algorithm for the Principal maps their observed history to a (distribution over) Agents to choose at each round, as well as a contract to offer. Once the Principal fixes an algorithm, this induces a $T$-round extensive form game amongst the Agents. We first show the existence of a particularly simple kind of equilibrium in this game.

\paragraph{Non-Responsive Pure Strategy  Equilibrium Existence} Extensive form games can have many equilibria that support a wide range of behavior through threats between players, which limits the predictive power of Nash equilibria on their own. We say that an equilibrium is \emph{non-responsive} if each player's strategy at each round depends on the realized sequence of states of nature, but \emph{not} on the history of actions of the other players. In Section \ref{sec:existence}, we prove a structural result: the game we study always admits a non-responsive pure-strategy Nash equilibrium. With this existence result in hand, in all of our subsequent results, we will give guarantees to the Principal under the assumption that the Agents are playing a non-responsive pure-strategy  equilibrium of the extensive form game induced by the Principal's algorithm. 

\paragraph{Regret to the Best Fixed Agent}
For our first algorithmic result (in Section \ref{sec:external}), we fix any contract $v(r)$ which is concave in $r$ (i.e. offers the Agent payment that is linear or has diminishing marginal returns of the Principal's reward), and give an extremely simple algorithm for the Principal. At each round $t$, the Principal selects an Agent using a monotone bandit algorithm with an external regret guarantee and offers contract $v$. The regret guarantee of the bandit algorithm immediately implies that the Principal has diminishing regret to the \emph{realized rewards} of the best Agent in hindsight. However, these can differ from the counterfactual rewards of giving the best Agent in hindsight a fixed contract, because the realized rewards are generated by actions the Agent played in a strategic game. What we show is that in the game induced by this mechanism, the rewards realized by each Agent are only higher than what they would have been counterfactually, had that Agent been given a fixed contract $v$ for all $T$ rounds. Thus the Principal has no regret to our counterfactual benchmark as well. This result hinges on the  \emph{monotonicity} of the bandit algorithm, which means  that unilaterally increasing the reward of a single Agent in a single round only increases the probability that this Agent is chosen in future rounds. Intuitively, the result is that in the competitive game induced by the mechanism, Agents are incentivized to spend more effort  than they would have otherwise, because the payoff they get is not just their one-round payoff, but also their continuation payoff over future rounds, which is increasing with current-round effort level. Monotone bandit algorithms  with external regret guarantees are already known to exist \cite{babaioff2009characterizing}.

\paragraph{Swap Regret and Limited Liability Contracts}
Our first result gives a mechanism that uses the same contract $v$ that its regret benchmark is defined with respect to. Thus if $v$ is not a limited liability contract, our mechanism also is not limited liability. This is problematic, because \emph{linear contracts} which are known to be robustly optimal in a variety of settings \citep{carroll2015robustness,dutting19,dutting2022combinatorial,castiglioni2021bayesian} and so would make ideal benchmarks, are not limited liability. Our second result (Section \ref{sec:swap}) gives an incentive-preserving reduction to limited liability contracts for any linear contract $v$. We use a construction similar to the one \cite{chassang2013calibrated} used to give a limited-liability mechanism for a single Agent. Informally, each Agent can open a ``tab'' with our Principal, which records the cumulative payment owed from the Principal to the agent after each round. The Principal will not pay Agents throughout, and instead pays each Agent in a lump sum at the end of the game according to their tab --- except that if the final tab indicates a negative payment from the Principal to that Agent. In this case, the Agent's debt is forgiven and the Principal pays them zero. A property of this setup is that if at the end of $T$ rounds, the linear contract $v$ would have resulted in a net transfer from the Principal to the Agent, the cumulative payments made under the tab are identical to what they would have been under the linear contract; if on the other hand the linear contract $v$ would have resulted in a net transfer from the Agent to the Principal, the cumulative payments made under the tab deviate from that of the linear contract. The tab enforces limited liability, but serves to limit the downside of the Agent, and so can be incentive-distorting. We show that if the Principal chooses Agents using an algorithm that has diminishing \emph{swap regret} (to the realized rewards), then the cumulative payment to the Agent under the tab is guaranteed to be non-negative (up to diminishing regret terms). This means that (again, up to low order regret terms) the utility for the Agent of playing any strategy is the same as it would have been had the mechanism been offering the linear contract $v$ rather than its limited liability proxy. We show that if the bandit algorithm used by the Principal obtains no swap-regret and is monotone (once again in the sense that unilaterally increasing the reward of a single Agent in a single round only increases the probability of choosing that Agent in future rounds), then in equilibrium, the Agents are incentivized to expend more effort than they would have had they been offered a guaranteed linear contract for all $T$ rounds, and hence the no (swap) regret guarantee of the Principal's mechanism over the realized rewards extends to a no (swap) regret guarantee to the counterfactual rewards that would be been obtained had the Principal picked the best fixed Agent in hindsight. We show that standard no swap-regret algorithms like \cite{blum2007external} are \emph{not} monotone in this sense --- but that the recent no swap-regret algorithms given by \cite{dagan2023external} and \cite{peng2023fast} are. These algorithms operate in the full information setting, but we give a monotonicity preserving reduction to the bandit setting. Thus we can use these algorithms in our reduction to promise the Principal no counterfactual (swap) regret to choosing the best Agent in hindsight (and offering them a fixed linear contract), while needing to offer only linear liability contracts.

\subsection{Related Work}

The origins of contract theory date to \cite{holmstrom1979moral} and \cite{grossman1992analysis}. This literature is too large to survey --- we refer the reader to \cite{bolton2004contract} for a textbook introduction.

Our work fits into a recent literature studying learning problems in contract theory. \cite{ho2014adaptive} study online contract design by approaching it as a bandit problem in which an unknown distribution over myopic Agents arrive and respond to an offered  contract by optimizing their expected utility with respect to a common prior. \cite{cohen2022learning} extend this to the case in which the Agent has bounded risk aversion. 
\cite{zhu2022sample} revisit this problem and characterize the sample complexity of online contract design in general (with nearly matching upper and lower bounds) and for the special case of linear contracts (with exactly matching upper and lower bounds). \cite{camara2020mechanisms} and \cite{collina2023efficient} study the problem of repeatedly contracting with a single non-myopic Agent in a setting in which there are no common prior assumptions, and the single Agent is assumed to follow various learning-theoretic behavioural assumptions (that their empirical play history satisfies various strengthenings of swap regret conditions). We distinguish ourselves from this line of work by studying the game induced by competition amongst multiple non-myopic Agents.

\cite{chassang2013calibrated} studies a repeated interaction between a Principal and a long-lived Agent, with a focus on the \emph{limited liability} problem. In the context of a single Agent and an outside option (and in a fully observable setting), \cite{chassang2013calibrated} gives a similar ``debt scheme''  to simulate the incentives of a linear contract using limited liability contracts. Our work extends \cite{chassang2013calibrated} in several ways: we study a competitive setting amongst many Agents, operate in a partially observable setting (we do not observe the outcome of Agents we did not select), and give counterfactual policy regret bounds (rather than regret bounds to the realized rewards). 

Our limited liability results hold for linear contracts, which have become focal in the contract theory literature because of their strong robustness properties.  \cite{carroll2015robustness} shows that linear contracts are minimax optimal for a Principal who knows \emph{some} but not \emph{all} of the Agent's actions.  \cite{dutting19} shows that if the Principal only knows the costs and expected rewards for each Agent action, then linear contracts are minimax optimal over the set of all reward distributions with the given expectation. \cite{dutting2022combinatorial} extends this robustness result to a combinatorial setting. \cite{dutting19} also show linear contracts are bounded approximations to optimal contracts, where the approximation factor can be bounded in terms of various quantities (e.g. the number of Agent actions, or the ratio of the largest to smallest reward, or the ratio of the largest to smallest cost, etc). 
\cite{castiglioni2021bayesian} studies linear contracts in Bayesian settings (when the Principal knows a distribution over types from which the Agent's type is drawn) and studies how well linear contracts can approximate optimal contracts.

Our results leverage the fact that in our mechanisms, the Principal is using a \emph{monotone} selection algorithm. The same kind of monotonicity was used by \cite{babaioff2009characterizing} in the context of using bandit algorithms to give truthful auctions. \cite{babaioff2009characterizing} used a monotone algorithm for obtaining no \emph{external} regret. We give a monotone algorithm for obtaining no \emph{swap-regret} by giving a  monotonicity preserving reduction to the full information swap regret algorithms recently given by \cite{dagan2023external,peng2023fast}, which we note are monotone. Prior swap regret algorithms (e.g. \cite{blum2007external}) even in the full information setting are (perhaps surprisingly) \emph{not} monotone, as we show.

\section{Preliminaries}
We study a multi-round, multi-Agent contracting problem. To formalize the model, we first focus on what happens in a single round.

\subsection{Single Round Contracting Model}
In each round, the Principal would would like to incentivize some Agent to complete a task. The Principal will offer a fixed contract to one of $k$ Agents, or pick the ``outside option" of not contracting at all. The task has multiple possible outcomes, each of which corresponds to a different value for the Principal. 

\begin{definition}[Outcome $o$ and Return $r(o)$]\label{def:outcome}
There are $m$ possible outcomes $o_{1}...o_{m}$, each of which is associated with a return for the Principal $r(o) \in [-1,1]$.
\end{definition}

\begin{definition}[Contract] \label{def:contract}
    A contract is a payment rule $v: \mathbb{R} \rightarrow \mathbb{R}$ mapping the Principal's return to a payment made to an Agent. 
\end{definition}
Throughout this paper, we will make the following natural assumption on contracts --- that higher reward to the Principal corresponds to higher payment to the Agent, and that payments exhibit (weakly) decreasing marginal returns.
\begin{assumption}
Contracts are concave and increasing functions.
\end{assumption}

If selected by the Principal, an Agent responds to the contract with an action $a$ in action set $\cA$. 

\begin{definition}[Agent Action and Cost Function] 
The action space for each Agent, $\cA = [0, 1]$, is all ``effort levels" that they can exert towards the Principal's task. Each effort level has an associated cost, possibly distinct for each Agent, given by a cost function $c_i: \cA \to \mathbb{R}$. 
\end{definition}

Throughout the paper, we will assume that Agent cost increases in effort and exhibits (weakly) increasing marginal cost to effort. Note that increasing marginal cost is consistent with an assumption of decreasing marginal \emph{return} to effort. 

\begin{assumption}
    Agent cost functions $c_i(\cdot)$ are increasing and convex in effort level.
\end{assumption}

There is also some hidden state of nature $y \in \mathcal{Y}$ in each round. The probability distribution over outcomes for an Agent is dictated by their chosen action and the state of nature. 

\begin{definition}[Outcome Distribution $\cD_{i, y, a}$] \label{def:outcome-dist}
There are $k \cdot m$ functions $p_{i,o}: \cA \times \cY \rightarrow [0,1]$, each associated with a particular Agent $i$ and outcome $o$, such that for all $y \in \cY$, $p_{i,o}(\cdot, y)$ is affine in its first argument. 
We will represent these functions as $p_{i, o} (a, y) = m_{i, o, y} \cdot a + b_{i, o, y}$.
Furthermore, for all Agents $i$, effort levels $a$, and all states of nature $y$, we have that \[ \sum_{j=1}^{m}p_{i,o_{j}}(a, y) = 1. \]

We assume that the probability that outcome $o$ occurs for Agent $i$ given effort level $a$ and state of nature $y$ is $p_{i,o}(a, y)$. 
By the constraints given above on the functions $p$, this is always a valid probability distribution.
We will refer to this distribution as $\cD_{i, y, a}$. Additionally, we will refer to the randomness used to sample from this distribution as $\cR_i$, so that fixing $a, y$, and $\cR_i$ also fixes the sampled outcome $o$.
Let $\bar{r}(i, a, y) = \E_{o \sim \cD_{i, y, a}} [r(o)]$ be the expected reward under the distribution of outcomes $\cD_{i, y, a}$. Let $\bar{v}(i, a, y) = \E_{o \sim \cD_{i, y, a}} [ v(r(o)) ]$ be the expected payment to Agent $i$ over the same distribution of outcomes.
\end{definition}

\begin{remark}
If the returns $r(\cdot)$ and the outcome distribution functions $p_{i, o}(\cdot)$ are such that effort level zero leads to zero return with probability one and the cost associated with zero effort is also zero, then we can model the ability of each Agent to leave the interaction at any time by playing effort level zero for all remaining time steps. This also ensures that Agents always have non-negative expected utility (no matter what their beliefs are) under their best response.
\end{remark}

After a selected Agent takes an action which induces an outcome $o$ and corresponding Principal reward $r(o)$, the Principal pays them according to the contract $v(r(o))$. 

\begin{definition}[Expected Single Round Agent Utility $u_{i}$]\label{def:single_round_utility_Agent}
The expected single round utility of Agent $i$ given that they are selected is $u_{i}(a,y) = \bar{v}(i, a, y) - c_i(a)$.
\end{definition}

The Principal's total utility that round is their reward minus this payment.

\begin{definition}[Expected Single Round Principal Utility $u_{p}$]\label{def:single_round_utility_Principal}
The expected single round utility of the Principal given that they select Agent $i$ playing action $a$ is $u_{p}(i,a,y) = \bar{r}(i, a,y) - \bar{v}(i, a,y)$. The single round utility if they select the outside option is $u_{p}(\emptyset, a, y) = 0$ for all $a \in \cA, y \in \cY$.
\end{definition}

Finally, we will make the natural assumption that if the Agent expends more effort, this only improves their expected return.

\begin{assumption}[Monotone Relationship between Effort and Expected Return] \label{assumption:monotone-effort-return}
For any Agent $i$, any $y \in \cY$ and any two effort levels $a_1, a_2$ such that $a_1 \geq a_2,$ the outcome values satisfy $\E_{o \sim \cD{i, y, a_1}}[r(o)] \geq \E_{o \sim \cD_{i, y, a_2}}[r(o)]$.
\end{assumption}


\subsection{Repeated Contracting Game}
In the repeated version of this game, the Principal can select which Agent to contract with each round based on the history of play. However, the Principal cannot observe the history of actions $\A$ or the states of nature $y$. The Principal only observes the history of realizations of the returns of the selected Agents\footnote{This is a crucial aspect of contract theory --- that the Principal must contract on the returns, and not directly on the actions of the Agent.}. 
Thus, we can think of the state of the Principal at round $t$ as being determined by the history of his selections and of returns realized from rounds $1$ to $t-1$.

\begin{definition}[Principal Transcript]
We call the history of the selected Agents in each round and their realized returns to the Principal the Principal transcript, denoted $\pi^{p, 1:t} = [(i^{t'}, r(a^{t'}, y^{t'}))]_{t' \in [t]}$.
The space of all possible Principal transcripts is denoted $\Pi^p$.
\end{definition}

Using only this transcript, the Principal must decide which Agent to contract with in each round. It will be useful to explicitly refer to the randomness used by the Principal in their selection mechanism when choosing an Agent in each round $t$. We refer to this randomness as $\cR^t$, and view the randomized selection of the Principal as a deterministic selection mechanism that takes $\cR^t$ as input.

\begin{definition}[Principal Selection Mechanism $f$] \label{def:Principal-selection-mech}
In each round $t$, the Principal uses a deterministic selection mechanism $f^t: \Pi^p \times \cR \to [k] \cup \emptyset$ to select an Agent (or the outside option). We will refer to the entire mechanism, consisting of all $T$ components, as $f$. The mechanism $f$ is known to the Agents. 
\end{definition}

In order to model Agent decisions, it is not enough that they know the Principal selection mechanism --- they must also have beliefs over the future realizations of the states of nature.  We allow each Agent to have a (potentially incorrect) belief about the future states of nature at each round, which may be a function of the previously realized states of nature. While Agent beliefs may differ, we assume that all Agents are aware of the beliefs of all other Agents at that round. The Principal does not have any knowledge of these beliefs.


\begin{definition}[Belief functions $B_i(\cdot)$] \label{def:belief-func}
Each Agent has a function $B^{t}_i: \cY^{t-1} \to \Delta(\cY^{T-t})$ for any $t \in [T]$  which is a mapping from any $t$-length prefix of states of nature to a distribution over the remaining $T-t$ states of nature. This represents Agent $i$'s belief at the start of round $t$ of the remaining $T-t$ states of nature given the $t-1$-length prefix of states of nature. Let $B^1_i = B^1_i(\emptyset)$ be the initial belief of Agent $i$ of the possible length $T$ sequence of natures states.  
Additionally, for any $t_2 \geq t_1 \geq t$, we let $B_i^{t, t_1:t_2}(\cdot)$ represent Agent $i$'s belief in round $t$ of the states of nature in between rounds $t_1$ and $t_2$, given some input of previous states of nature. Furthermore, $B^{t'}(y^{1:t})$ represents the set of beliefs of all Agents at round $t'$ given $y^{1:t}$. Agents are aware of each others' belief functions. 
\end{definition}

\begin{assumption}[Full Support Beliefs]
    We assume that the beliefs of each agent have full support over all possible sequences of states of nature, that is $\forall i \in [k], B^{1}_{i}$ has full support. 
    We make no assumption on the magnitude of the minimum probability, simply that it is non-zero on each sequence.
\end{assumption}

The selection function $f$ and set of Agent belief functions $B$ together define a $k$-player extensive form game amongst the Agents, in which each Agent aims to maximize their utility over the entire time horizon of the interaction. 

The Agents are more informed than the Principal, so they \emph{do} have access to the history of states of nature and the actions of other Agents. Therefore, their state at each round is more complex, and determined by the following transcript. 

\begin{definition}[Agent Transcript]
We call the history of the realized states of nature $y^{1}...y^{t-1}$, the history of the random bits used to sample the selected Agent's outcome $\cR^1_{i^1}, \ldots, \cR^{t-1}_{i^{t-1}}$, the history random bits used by the Principal $\cR^1, \ldots, \cR^{t-1}$, and the history of actions of the selected Agents $a^1, \ldots, a^{t-1}$ an Agent transcript,
denoted $\pi^{1:t} = [(y^{t'}, \cR^{t'}_{i^{t'}}, \cR^{t'}, a^{t'})]_{t' \in [t]}$. 
We refer to a single element $t'$ of an Agent transcript as $\pi^{t'}$.
The space of all possible Agent transcripts is denoted $\Pi$. 
\end{definition}

Note that the Agent transcript contains strictly more information than the Principal transcript in any round. 
The action space of each Agent $i$ is a policy mapping from Agent transcripts to actions in each round. 

\begin{definition}[Agent Policy]
An Agent policy $p_i: \Pi \to \Delta(\cA)$ is a mapping from any $t$-length transcript $\pi^t$ to a distribution over actions.
We will write $p^{t}_{i}(\pi^{t-1})$ as the distribution over actions output by Agent $i$ at round $t$. 
Let $\cP$ be the set of all Agent policies. 
\end{definition}

\begin{game} 
\caption{Agent Selection Game}
\label{setting:Agent-selection}
    \begin{algorithmic}
        \STATE {\bf Input} fixed contract $v$, time horizon $T$
        \STATE Agent cost functions $c_i(\cdot)$ and belief functions $B_i(\cdot)$ are common knowledge amongst all Agents. 
        \STATE Principal fixes their selection mechanism $f$, which is common knowledge to all Agents.
        \STATE Agent transcript is initialized to transcript $\pi = \{\}$
        \STATE Principal transcript is initialized to $\pi^p = \{\}$
        \FOR{$t \in [T]$}
            \STATE Principal selects Agent $f^t(\pi^{p, t-1}, \cR^t) = i^t$ to whom to award contract
            \STATE Agent $i^t$ takes action $a^t$, observed by all Agents, at a cost of $c_{i^t}(a^t)$
            \STATE Agents observe the state of nature $y^t$
            \STATE Agents observe the outcome $o_{i^t}^t \sim \cD_{i^t, y^t, a^t}$ and the randomness $\cR^t_i$ used to sample it 
            \STATE Principal receives the return from realized outcome $r(o_{i^t}^t)$
            \STATE Agent $i^t$ receives payment $v(r(o_{i^t}^t))$
            \STATE Principal publishes randomness $\cR^t$ used in round $t$ 
            \STATE Principal transcript is updated with $(i^t, r(o_{i^t}^t))$
            \STATE Agent transcript is updated with $(y^t, \cR^t_{i^t}, \cR^t, a^t_{i^t})$
        \ENDFOR
    \end{algorithmic}
\end{game}

Taken together, a Principal selection mechanism $f$, a set of Agent policies $p$, and a distribution over states of nature $\hat{y}$ fully define a distribution over transcripts of the extensive form game. 
When useful, we will consider the function which generates a distribution of transcripts given these inputs. This function will generate the distribution of \emph{all} possible transcripts that may occur based on a distribution over the future states of nature which is taken as an input. Namely, the additional randomness over which it is generating a distribution is the randomness in the other components of the transcript, which the Agent beliefs do not concern themselves with - the randomness used by the Principal in their selection mechanism and the randomness used to sample the outcome from the selected Agent's outcome distribution in each round. 

\begin{definition}[Transcript Generating Function] \label{def:transcript-generating-func}
Given a Principal with selection mechanism $f$ and Agents policies $p$,
the function $g_{f, p}: \Delta(\cY)^T \to \Delta(\Pi)$,
takes as input a distribution over $T$ states of nature $\cY$
and outputs a distribution over transcripts.
Given some prefix of the transcript of length $t$, we can also write $g^{1:t'}_{f,p}: \Pi^{1:t} \times \Delta(\cY)^{t'-t} \to \Delta(\Pi^{1:t'})$ for any $t' \geq t$
to denote the distribution over transcripts until round $t'$. When useful, we will also allow this function to take as input a deterministic sequence of nature states.
\end{definition}

Now, we are ready to describe the utility of each Agent $i$ in the game defined by $f$ and belief function $B_i(\cdot)$. We will describe this for every \emph{subgame} of the game, defined by an Agent transcript prefix. While the equilibrium guarantees of our mechanisms do not require subgame perfection, it will be useful for us in our analysis to think about the payoffs in particular subgames.

\begin{definition}[Agent Subgame Utility] \label{def:utility}
    At round $t$, given a previously realized transcript $\bar{\pi}^{1:t-1}$ with states of nature of $\bar{y}^{1:t-1}$, Agent $i$'s payoff function for policy $p_{i}$ is defined as
    \begin{align*}
        \mathcal{U}_{i}(t, \bar{\pi}^{1:t-1}, p_{j \neq i},p_{i}) = 
       \mathbb{E}_{\pi^{1:T} \sim g^{1:T}_{f, p}(\bar{\pi}^{1:t-1},B^{t, t:T}_i (\bar{y}^{1:t-1}))} \left[\sum_{t' = t}^{T} \mathbb{I}[(f^t(\pi^{p, 1:t'-1}, \cR^t) = i)] \cdot u_i(p^{t'}_{i}(\pi^{1:t'-1}), y^{t'}) \right].
    \end{align*}
\end{definition}

In some cases, it will also be useful to refer to Agent $i$'s realized utility over a transcript prefix $\pi^{1:t}$, which includes information of all relevant elements the Agent's utility -- namely, the Principal selection in each round and the outcome of the selected Agent action -- as $\mathcal{U}_{i}(\pi^{1:t})$.

\begin{definition}[Principal Utility] \label{def:Principal-utility}
    The Principal's payoff function for selection mechanism $f$, Agent policies $p^*$, and a distribution over state of nature sequences $y^{1:T}$ is defined as 
    \begin{align*}
        \mathcal{U}_{p}(p^{*},f,y^{1:T}) = 
       \mathbb{E}_{\pi^{1:T} \sim g^{1:T}_{f,p^*}(y^{1:T})} \left[\sum_{t' = 1}^{T} \mathbb{I}[(f^t(\pi^{p,1:t'-1}, \cR^t) = i)] \cdot u_p(i,p^*_{i}(\pi^{1:t'-1}), y^{t'}) \right].
    \end{align*}

    We also let $\mathcal{U}_{p}(\pi^{p})$ denote the utility of the Principal over some Principal transcript $\pi^{p}$.
\end{definition}




\begin{definition}[Agent Best Response]
Given a set of belief functions $B$ and other Agent policies $p_{j \neq i}$, Agent $i$'s policy  $p^{*}_{i}$ is a best response if
$$p^{*}_{i} \in \argmax_{p_{i} \in \mathcal{P}}(\mathcal{U}_{i}(1, \emptyset, p_{j \neq i},p_{i})).$$
\end{definition}

\begin{definition}[Agent Nash Equilibrium]
Given a set of belief functions $B$ and a Principal selection mechanism $f$, Agent policies $p^{*}_{1...k}$ are a Nash equilibrium if, for each Agent $i$, $p^{*}_{i}$ is a best response to $p^{*}_{j \neq i}$.
\end{definition}

\subsection{Non-Responsive Policies}

Since effort levels are continuous, the set of policies $\mathcal{P}$ is an infinite set, and in fact each $p \in \cP$ is also an infinite-dimensional object. Thus, it is not immediately obvious that the game we study (Game \ref{setting:Agent-selection}) even admits a Nash equilibrium, and if so what properties it might have. Additionally, Nash equilibria in extensive form games can lack predictive power and support a wide range of behaviors through threats and collusion. For these reasons we will want to prove the existence not only of Nash equilibria, but of \emph{non-responsive} Nash equilibria, in which Agent policies do not in fact depend on the history of actions of the other players, or the history of choices of the mechanism (which themselves depend on the history of actions)---but rather only on the history of nature states. 

\begin{definition}[Non-responsive Equilibrium]
    Given a set of belief functions $B$ and a Principal selection function $f$, a set of Agent policies   
    $p^{*}_{1...k}$ is a non-responsive equilibrium if it is an equilibrium and all policies in the set are non-responsive.
\end{definition}

\section{Existence of Non-Responsive Equilibrium} \label{sec:Agent-selection}

In this section we show that the game induced by the Agent selection interaction described in Game \ref{setting:Agent-selection} admits a pure strategy non-responsive Nash equilibrium amongst the Agents. We will refer to Game \ref{setting:Agent-selection} as the ``general game," to distinguish it from a related game we define in this section. 

\begin{theorem} \label{thm:pure-nonresponsive}
There exists a pure non-responsive Nash equilibrium of the Agent Selection Game (Game \ref{setting:Agent-selection}).
\end{theorem}

The remainder of this section is devoted to the sequence of arguments that result in the proof of Theorem \ref{thm:pure-nonresponsive}. The outline of the proof is as follows:
\begin{itemize}
    \item First, we define the \emph{restricted game}, in which Agents have a restricted strategy space and are  only able to employ \emph{non-responsive} strategies --- i.e. strategies that are independent of the actions of other Agents. 
    \item Next, we apply the result of \citet{glicksberg} to show that the restricted game admits a mixed strategy Nash equilibrium (Lemma \ref{lem:glicksberg-application}).
    \item We show that we can purify any mixed strategy Nash equilibrium of the restricted game while maintaining its equilibrium properties; this establishes that the restricted game admits a pure strategy Nash equilibrium (Lemma \ref{lem:pure}).
    \item Finally we show that any pure strategy Nash equilibrium of the restricted game must also be a pure strategy Nash equilibrium of the general game. This establishes that the general game admits a pure strategy Nash equilibrium \emph{in which all of the players are employing non-responsive strategies} (Theorem \ref{thm:general-eq}).
\end{itemize}

\subsection{Restricted Contracting Game}

To execute the first step of the outline, we define the restricted contracting game in which Agents are restricted to playing non-responsive strategies. Unlike the general game, the strategy space in the restricted contracting game is finite-dimensional, and so it is easier to argue equilibrium existence; and of course, equilibria in the restricted contracting game are non-responsive by definition. Once we establish the existence of pure-strategy Nash equilibria in the restricted game, we will show how to ``lift'' them to pure strategy non-responsive equilibria of the general game. Agent policies in the restricted game will be functions of \emph{restricted transcripts}. Restricted transcripts contain a history of the states of nature and randomness used in the playout, but crucially do \emph{not} contain records of Agent actions or functions of Agent actions (like outcomes or Principal selections). The result is that policies that are defined as functions of restricted transcripts will be independent of the actions of other players.

\begin{definition}[Restricted Transcript] 
We call the collection of the history of the realized states of nature $y^1, \ldots, y^{t-1}$, the history of the random bits used to sample the selected Agent's outcome in each round $\cR^t_{i^t}$, and the history of the random bits used by Principal $\cR^1, \ldots, \cR^{t-1}$ a restricted transcript, denoted $\pi^{r,1:t} = [(y^{t'}, \cR^{t'}_{i^{t'}}, \cR^{t'})]_{t' \in [t]}.$ 
We refer to a single element $t'$ of a restricted transcript as $\pi^{r, t'}$.
The space of all possible restricted transcripts is denoted $\Pi^r$.
\end{definition}

We will refer to Agent $i$'s utility over the restricted transcript prefix $\pi^{r, 1:t}$ by $\mathcal{U}_{i}(\pi^{r, 1:t})$. 

\begin{definition}[Restricted Agent Policy] \label{def:restricted-policy}
A restricted Agent policy $p_i: \Pi^r \to \Delta(\cA)$ for an Agent $i$ is a mapping from any $t$-length restricted transcript to a distribution over actions $\Delta(\cA)$.
We will write $p^{t}_{i}(\pi^{r, t-1})$ as the distribution over actions output by Agent $i$ at round $t$. 
We will call the set of all restricted policies $\cP^r$.
\end{definition}

Note that all restricted Agent policies are \emph{non-responsive} by construction - they do not depend on other Agents' actions. 
In the general game, we refer to any policy that can be written as a function of the restricted transcript  a ``non-responsive'' strategy.

\begin{definition}[Restricted Transcript Generating Function]
Given a Principal selection mechanism $f$, the function $g^r_f: \Delta(\cY^T) \to \Delta(\Pi^r)$, takes as input a distribution over a sequence of $T$ states of nature $\cY$ and outputs a distribution over restricted transcripts.
Given some prefix of the restricted transcript of length $t$, we can also write $g^{r, 1:t'}_{f}: \Pi^{1:t} \times \Delta(\cY)^{t'-t} \to \Delta(\Pi^{1:t'})$ for any $t' \geq t$ to denote the distribution over transcripts through round $t'$. 
\end{definition}

\subsection{Equilibrium Existence} \label{sec:existence}

We begin by applying the following result to the restricted game. 

\begin{theorem}[\cite{glicksberg}] \label{thm:glicksberg}
    Every continuous and compact game has a mixed strategy Nash equilibrium. 
\end{theorem}

To do this, we will first show a result that will be useful throughout this paper for relating the local changes in an Agent's policy to their global utility in the restricted game.

\begin{lemma} \label{lem:subgame-equivalence}
Fix a selection function $f$ and a strategy profile $p$ of Agents in the restricted game. Let $\hat{p}_i$ be a policy for Agent $i$ that is identical to $p_i$ except for the action in round $t$ for some transcript prefix $\pi^{r, 1:t-1}:$ so $p^t_i(\pi^{r, 1:t-1}) \neq \hat{p}^t_i(\pi^{r, 1:t-1})$.
Then, the difference in Agent $i$'s utilities between these two policies can be written as:
\begin{align*} \mathcal{U}_{i}(1,\emptyset,p_{j \neq i},\hat{p}_{i}) &- \mathcal{U}_{i}(1,\emptyset,p_{j \neq i},p_{i}) \\
& \quad \quad \quad =\Pr[g^{r, 1:t-1}_{f}(B^1_i) = \pi^{r,1:t-1}] \cdot \left(\mathcal{U}_{i}(t, \pi^{r,1:t-1}, p_{j \neq i},\hat{p}_{i}) - \mathcal{U}_{i}(t, \pi^{r,1:t-1}, p_{j \neq i},p_{i}) \right).
\end{align*}
\end{lemma}

The proof of Lemma~\ref{lem:subgame-equivalence} is deferred to Appendix~\ref{app:subgame-equivalence}.

\begin{lemma} \label{lem:glicksberg-application}
The restricted game has a mixed strategy Nash equilibrium.
\end{lemma}

\begin{proof}
We show that the restricted game is continuous and compact and then apply Theorem \ref{thm:glicksberg}.

Each player's strategy space is all mappings from restricted transcripts to distributions over effort levels. 
As we assume that the coin flips used by the Principal, the randomness used to sample an outcome, and space of states of nature are finite, there are a finite number, say $d$, of restricted transcripts, so this is a finite, $d$-dimensional game. Each Agent's strategy in the restricted game can be represented as a real-valued $d$-dimensional vector lying in a compact space $[0, 1]^d$, as effort levels are defined to lie in $[0, 1]$. 

Next, we show that the utility functions are continuous in Agent $i$'s policy $p_i$. Consider a fixed restricted transcript prefix $\pi^{r, 1:t-1}$. We show that the change in Agent $i$'s expected payoff is continuous in $p^t_i(\pi^{r, 1:t-1})$. Consider perturbing the action played by $p^t_i(\pi^{r, 1:t-1})$ to some new action $a'$. Let $\hat{p}_{i}$ represent this perturbed policy. This corresponds to changing a single coordinate in the $d$-dimensional vector representing Agent $i$'s policy. As this change is confined to a single subgame, by Lemma~\ref{lem:subgame-equivalence}, the only change to the utility is in the expression $\mathcal{U}_{i}(t, \pi^{r,1:t-1}, p_{j \neq i},\hat{p}_{i})$. So, it is sufficient to show that the change in this expression is continuous.

Recall that an Agent's utility in this subgame $\bar{\pi}^{r,1:t-1}$ can be written as
    \begin{align*}
        \mathcal{U}_{i}(t, \bar{\pi}^{r, 1:t-1}, p_{j \neq i},&p_{i}) = 
       \mathbb{E}_{\pi^{r, 1:T} \sim g^{r, 1:T}_{f}(\bar{\pi}^{r, 1:t-1},B^t_{i}(\bar{y}^{1:t-1}))} \left[\sum_{t' = t}^{T} \mathbb{I}[(f^t(\pi^{p, 1:t'-1}, \cR^t) = i)] \cdot u_i(p^{t'}_{i}(\pi^{r, 1:t'-1}), y^{t'}) \right]  \\
        & =\underbrace{\mathbb{E}_{\pi^{r, 1:t} \sim g^{r, 1:t}_f(\bar{\pi}^{r, 1:t-1}, B^{t, t}_{i}(\bar{y}^{1:t-1}))} 
        \left[\mathbb{I}[(f^t(\pi^{p, 1:t-1}, \cR^t) = i)] \cdot u_i(p^{t}_{i}(\pi^{r, 1:t-1}), y^{t}) 
        \right]}_{\text{immediate payoff}} \\
        & \quad + \underbrace{\mathbb{E}_{\pi^{r, t+1:T} \sim g^{r, 1:T}_{f}(\bar{\pi}^{r, 1:t-1},B^t_{i}(\bar{y}^{1:t-1}))} \left[\sum_{t' = t+1}^{T} \mathbb{I}[(f^t(\pi^{p, 1:t'-1}, \cR^t) = i)] \cdot u_i(p^{t'}_{i}(\pi^{r, 1:t'-1}), y^{t'}) \right]}_{\text{continuation payoff}}.
    \end{align*}
 
The immediate payoff changes as a result of this perturbation only in the single round utility the Agent sees: $u_i(p^t_i(\pi^{r, 1:t-1}), y^t) = \mathbb{E}_{a \sim p^t_i(\pi^{r, 1:t-1})} [ \bar{v}(i, a, y^t) - c_i(a) ]$. The cost function is, by assumption, convex on a closed set, and thus continuous. It remains to be shown that the expected payment to the Agent also continuous in the effort level. The expected payment is
\begin{align*}
    \mathbb{E}_{a \sim p^t_i(\pi^{r, 1:t-1})} [ \bar{v}(i, a, y^t)] = \mathbb{E}_{a \sim p^t_i(\pi^{r, 1:t-1})} [ \E_{o \sim \cD_{i, y, a}} [ v(r(o)) ] ].
\end{align*}
This is continuous in the effort level because the expected outcome is continuous in the effort level and payment is a post-processing of the outcome via the return $r(\cdot)$ and the contract $v(\cdot)$.
The expected outcome is continuous in the effort level because the probability of each outcome, for a fixed state of nature $y$, is affine in the effort level:
\begin{align*}
    \E_{o \sim \cD_{i, y, a}} [o] = \sum_{j \in [m]} p_{i, o_m}(a, y). \cdot o_m.
\end{align*}
Therefore, the immediate payoff for Agent $i$ is continuous in their action. 

Next, observe that the continuation payoff also changes only as a function of how the perturbation affects the return Agent $i$ realizes for the Principal. This is because the effect in the continuation payoff of an Agent's action in round $t$ is simply via the probability that they are selected by the Principal's selection mechanism in future rounds, which takes as input the information of previously realized returns. Thus, by the same argument above, we also have that the continuation payoff is continuous in an Agent's action.
Therefore, we have established that utility functions are continuous in Agent policies. 

As this game is continuous and compact, by Theorem \ref{thm:glicksberg}, it must have a mixed strategy Nash equilibrium.
\end{proof}

Next, we show that without loss of generality, any mixed strategy Nash equilibrium in the restricted game can be converted into a pure strategy Nash equilibrium. Towards that end, the next lemma shows that if we change any Agent's mixed strategy at any single round into a pure strategy by having them deterministically play the expected action under their mixed strategy, then the distribution over restricted transcripts is unchanged. 

\begin{lemma} \label{lem:p_tran_dist} 

Fix a Principal selection algorithm $f$ and a set of non-responsive, possibly mixed, Agent strategies $p^*$.
Let $i$ be an Agent who is playing a mixed strategy under $p^*_i$ in some round $t'$ for a restricted Agent transcript $\pi^{r, 1:t'-1}$. Let $p'$ be the strategy profile that is identical to $p^*$ for all Agents $j \neq i$ and differs for Agent $i$ only for transcript $\pi^{r, 1:t'-1}$ where $p'^{, t'}_i ( \pi^{r, t'-1}) = \E_{a \sim p^{*, t'}_i (\pi^{r, t'-1})} [a]$.
Then, the distribution over future Principal transcripts $\pi^{p, t':T}$ induced by strategy profile $p^*$ is  identical to the distribution over future Principal transcripts $\pi^{p, t':T}$ induced by strategy profile $p'$. 
\end{lemma} 
\begin{proof}
Observe that before time $t'$, there is no difference in the game, and thus the distributions of transcripts until $t'$ will be identical. 
Since all Agents are non-responsive, their actions are not a function of other Agent policies. Thus, the distributions of actions of all Agents except Agent $i$ will remain the same in both strategy profiles $p^*$ and $p'$. This also means that for all Agents $j \neq i$, the outcomes distributions they induce are the same in both worlds. Furthermore, since the randomness used by the Principal is fixed between both worlds, we know the Agent selections are the same.  

It remains to be shown that the outcome distributions for Agent $i$ under both strategy profiles $p^*_i$ and $p'_i$ are the same. 
At time $t$, under strategy $p^*_i$ for Agent $i$, we have outcome distribution $\bar{o}(p^*_i(\pi^{r, 1:t-1}),y^t)$.

We will show that for Agent $i$, the probability for any outcome is identical under strategy $p^*_i$ and $p'_i$. Recall that per Definition $\ref{def:outcome-dist}$, the probability that an outcome $o$ occurs for an Agent $i$ given an effort level $a$ and state of nature $y$ is an affine function $p_{i, o} ( a, y )$, which we represent as $p_{i, o} (a, y) = m_{i, o, y} \cdot a + b_{i, o, y}$.

Let $P^{*}(a_{j}) = Pr_{a \sim p^{*}_i(\pi^{r, 1:t-1})}[a_{j}]$ and let $P'(a_{j}) = Pr_{a \sim p'_i(\pi^{r, 1:t-1})}[a_{j}]$.

Now, we can write the probability of an outcome $o_s$ for Agent $i$ occurring under policy $p^*_i$ as
\begin{align*}
\int_{a=0}^{1} & p_{o_{s},y^{t}}(a) \ dP^{*}(a) \\
& = \int_{a=0}^{1} (m_{i, o_s, y^t} \cdot a + b_{i, o_s, y^t} ) \ dP^{*}(a) \\
& = \int_{a=0}^{1} (m_{i, o_s, y^t} \cdot a ) \ dP^{*}(a) + \int_{a=0}^{1} b_{i, o_s, y^t}  \ dP^{*}(a) \\
& = m_{i, o_s, y^t} \cdot \int_{a=0}^{1} a  \ dP^{*}(a) + b_{i, o_s, y^t} \\
& = m_{i, o_s, y^t} \cdot \mathbb{E}_{a \sim p^{*}_i(\pi^{r, 1:t-1})}[a] + b_{i, o_s, y^t} \\
& = m_{i, o_s, y^t} \cdot \mathbb{E}_{a \sim p'_i(\pi^{r, 1:t-1})}[a] + b_{i, o_s, y^t} \tag{by the fact that the expected value of $a$ is the same under $p^{*}$ and $p'$} \\
& = \int_{a=0}^{1} p_{o_{s},y^{t}}(a) \ dP'(a) \tag{by reversing the above steps with $P'$}
\end{align*}

Thus, any $p'_i(\pi^{r, 1:t-1})$ with the same expected effort level will give an equal distribution over outcomes. In particular, this is true for $p'_i$ that deterministically plays the expected effort level of $p^*_i$.
Therefore, we have shown that under both strategy profiles $p^*$ and $p'$, the distribution over Principal transcripts is identical.
\end{proof}

The next lemma shows that if we purify a single Agent's equilibrium strategy by having them play the expectation of their mixed strategy, then their opponents are still playing best responses.

\begin{lemma} \label{lem:purify-br}
Let $p^*$ be a strategy profile of an equilibrium of the restricted game. Let $p^*_i$ be the policy for an Agent $i$ who is playing a mixed strategy at some round $t$ for transcript prefix $\pi^{r, 1:t-1}$. Let $p'_i$ be an a policy for Agent $i$ that is almost identical to $p^*_i$, except plays a pure strategy at round $t$, specifically the mean effort level of the distribution of $p^*_i(\pi^{r, 1:t-1})$: that is, $p'^{, t'}_i ( \pi^{r, t'-1}) = \E_{a \sim p^{*, t'}_i (\pi^{r, t'-1})} [a]$. 
Then, for all Agents $j \neq i$, it is the case that $p^*_j$ is a best response to $p^{*}_{:-i,-j},p'_{i}$:
\begin{align*}
    p^{*}_{j} \in \argmax_{p_{j} \in \mathcal{P}^r}(\mathcal{U}_{j}(1, \emptyset, \{p^{*}_{:-i,-j},p'_{i}\}, p_j).
\end{align*}
\end{lemma}
\begin{proof}

Consider some Agent $j$, where $j \neq i$. Fix some prefix $\bar{\pi}^{r, 1:t-1}$ of the restricted game including states of nature $\bar{y}^{1:t-1}$.
Finally,  let $\tilde{\pi}^{p,1:t}$ denote the Principal transcript given an arbitrary restricted Agent transcript $\pi^{r, t}$ and the Agent policies $(p^*_{:-i,-j},p_j,p'_i)$, and define $\pi^{*,p,1:t}$ as the Principal transcript given an arbitrary restricted Agent transcript $\pi^{r, 1:t}$ and the Agent policies $(p^*_{:-i,-j},p_j,p^{*}_i)$.

Then, the expected payoff of the Agent $j$ over their own belief given any policy $p_{j}$ against the policies $p^*_{:-i,-j},p_j,p'_i$ is:
\begin{align*}
\mathbb{E}_{\pi \sim g^{r, 1:T}_{f}(\bar{\pi}^{r, 1:t-1},B^t_{j}(\bar{y}^{1:t-1}))} \left[\sum_{t' = t}^{T} \mathbb{I}[f^t(\tilde{\pi}^{p, 1:t'-1}, \cR^t) = j] \cdot u_j(p^{t'}_{j}(\pi^{r, 1:t'-1}), y^{t'}) \right].
\end{align*}

While their payoff using $p_{j}$ against $p^*$ is:

\begin{align*}
    \mathbb{E}_{\pi \sim g^{r, 1:T}_{f}(\bar{\pi}^{r, 1:t-1},B^t_{j}(\bar{y}^{1:t-1}))} \left[\sum_{t' = t}^{T} \mathbb{I}[f^{t}(\pi^{*, p, 1:t'-1}, \cR^t) = j] \cdot u_j(p^{t'}_{j}(\pi^{r, 1:t'-1}), y^{t'}) \right].
\end{align*}

Recall that the restricted transcripts exclude information about the Agent actions, and therefore both these expressions are evaluated over the exact same restricted transcript.

Also note that, by Lemma~\ref{lem:p_tran_dist}, it is the case that $\tilde{\pi}^{p, 1:t'}$ and $\pi^{*, p, 1:t'}$ are identically distributed. Therefore, the two payoff expressions above are the same, and the utility of Agent $j$ is the same.


Now, assume for contradiction that there is some policy $p_{j}$ which strictly outperforms $p^{*}_{j}$ against $p^*_{:-i,-j},p'_i$. Given the equivalency of the utility functions, this implies that $p_j$ strictly outperforms $p^*_j$ against $p^*_{:-j}$, leading to a contradiction.
\end{proof}

By iteratively applying this kind of strategy purification, we can establish the existence of a pure strategy Nash equilibrium in the restricted game.

\begin{lemma}\label{lem:pure}
There exists a pure Nash equilibrium of the restricted game.
\end{lemma}

The proof of Lemma~\ref{lem:pure} is deferred to Appendix~\ref{app:pure}.

The next step is to ``lift'' a pure strategy equilibrium of the restricted game to a pure strategy (non-responsive) equilibrium of the general game. The difference between a strategy of the restricted game and a strategy of the general game is that a strategy of the general game can depend on a richer transcript which includes information about the actions of other Agents. We first observe, however, that when all Agents are playing deterministic, non-responsive strategies, the richer ``full'' transcript of the general game in fact contains no additional information beyond the transcript of the restricted game. 
 
\begin{lemma} 
Whenever all Agents are playing deterministic policies, the full transcript $\pi^{1:t}$ for an arbitrary round $t$ can be computed as a deterministic function of the restricted transcript $\pi^{r, 1:t}$. \label{lem:det-policies} 
\end{lemma}
\begin{proof}
Consider a deterministic Agent strategy profile $p$ and a restricted transcript $\pi^{r, 1:t}$. We will show that we can compute the full Agent transcript $\pi^{1:t}$ using only $\pi^{r, 1:t}$ and information known to all Agents, namely Agent policies and belief functions $B_i(\cdot)$ for all $i \in [k]$. Note that the only difference between the transcripts is that the full Agent transcript contains the the action of the selected Agent in each round, so to show the claim we must simply show that in any round, any Agent can compute which Agent was selected by the Principal and what action the selected Agent took. 

We proceed via induction. In round $t = 1$, any Agent can compute the action $a^1$ since the Agent transcript is empty, all Agent strategies are deterministic, and both Agent belief functions and the initial state of the Principal's selection mechanism are known to all Agents.

Now, consider some arbitrary round $t > 1$. Assume that $a^{t-1}$ can be derived from $\pi^{r, t-1}$. We will show that $a^t$ can be derived from $\pi^{r, t}$. 
First, observe that an arbitrary Agent $j$ can compute which Agent will be selected by the Principal in round $t$. The Principal's selection mechanism takes as input the Principal transcript, which includes the historical selections of Agents and their realized outcomes. By the inductive hypothesis, all Agents know the history of actions of selected Agents through round $t-1$. Additionally, they know, by the restricted transcript $\pi^{r, t-1}$, the randomness $\cR^t_{i^t}$ used to sample the selected Agent's outcome in each round $t$. Both of these together exactly determine the state of the Principal's selection mechanism at round $t$. Thus, using $\cR^t$, which is contained in $\pi^{r, t}$, Agents know which Agent was selected in round $t$. 

Secondly, observe that the selected Agent's action $a^t$ is a known, deterministic function that depends on $\pi^{t-1}$ and $B_i(\cdot)$. Therefore, once again, applying the inductive hypothesis, Agents can compute the selected Agent's action $a^t$. 
\end{proof}

To reason about utilities in the general game, we will show an extension of Lemma~\ref{lem:subgame-equivalence} to this setting.

\begin{lemma} \label{lem:subgame-equivalence-general}
In the general game, fix a selection function $f$ and a set of (possibly responsive) policies of other Agents $p_{j \neq i}$. Consider two (possibly responsive) policies for Agent $i$, denoted $\hat{p}_{i}$ and $p_{i}$, which only differ in a given transcript prefix $\pi^{1:t-1}$. Then, 
\begin{align*}
    \mathcal{U}_{i}(1,\emptyset,p_{j \neq i},\hat{p}_{i}) - \mathcal{U}_{i}(1,\emptyset,p_{j \neq i},p_{i}) = \Pr[g_{f}(B^1_i) = \pi^{1:t-1}] \cdot \left(\mathcal{U}_{i}(t, \pi^{1:t-1}, p_{j \neq i},\hat{p}_{i}) - \mathcal{U}_{i}(t, \pi^{1:t-1}, p_{j \neq i},p_{i}) \right).
\end{align*}
\end{lemma}
\begin{proof}
This proof follows the same argument as Lemma \ref{lem:subgame-equivalence}.
\end{proof}

Finally, combining these lemmas lets us establish that pure strategy Nash equilibria of the restricted game are in fact also pure strategy Nash equilibria of the general game.

\begin{lemma}
Every pure Nash equilibrium of the restricted game is also a pure Nash equilibrium of the general game. \label{lem:extend}
\end{lemma}

\begin{proof}
Consider any pure strategy equilibrium $p$ in the restricted game. Each Agent is playing a policy in their best response set over all restricted pure strategies. Say Agent $i$ plays policy $p_i$ in this equilibrium.
We will show that in the general game, for any Agent $i$, there is no beneficial \emph{responsive} deviation from their restricted policy $p_i$, assuming all other Agents are playing $p_{:-i}$. Suppose for the sake of contradiction that there is a beneficial  deviation for Agent $i$ to responsive strategy $q_i$ in the general game. 




We can show that in fact, this is a valid strategy to play in the restricted game - i.e. that it can be implemented as a non-responsive strategy - as follows. 
By Lemma \ref{lem:det-policies}, when Agents are playing deterministic, non-responsive policies -- as all policies in the profile $p$ are -- any Agent $i$ can in round $t$ compute, using only $\pi^{r, 1:t-1}$, the full transcript $\pi^{1:t-1}$, which is the exact information they would have access to in the general game. Therefore, in fact, $q_i$ is also implementable as a non-responsive strategy $\tilde{q}_i$, as it can be implemented only as a function only of the restricted transcript as follows: in round $t$, let $\tilde{q}^t_i(\pi^{r, t-1})$ first compute the full transcript from the restricted transcript and then output the same action as $q^t_i(\pi^{1:t-1})$. 
Since the policy $q_i$ can be written as an equivalent policy which is only a function of the restricted transcript, it is by definition a non-responsive policy. 
Thus, it would necessarily also be a beneficial deviation from $p^t_i$ in the restricted game. This gives us a contradiction to the fact that $p$ was a pure strategy equilibrium of the restricted game.

We can see this as follows. 
\begin{align*}
    \max_{q_i \in \cP} \  &\mathcal{U}_{i}(1, \emptyset, p_{j \neq i},q_i) - \mathcal{U}_{i}(1, \emptyset, p_{j \neq i},p_{i})  > 0 \tag{suppose that $q_i$ is a beneficial responsive deviation from $p_i$} \\
   & \rightarrow \max_{q \in \cP}\ \mathcal{U}_{i}(1, \emptyset, p_{j \neq i},\tilde{q}_i) - \mathcal{U}_{i}(1, \emptyset, p_{j \neq i},p_{i})  > 0 \tag{ $\tilde{q}$ can be implemented as an equivalent restricted policy $\tilde{q}_i$} \\ 
   & \rightarrow \max_{q \in \cP^r}\ \mathcal{U}_{i}(1, \emptyset, p_{j \neq i},\tilde{q}_i) - \mathcal{U}_{i}(1, \emptyset, p_{j \neq i},p_{i})  > 0 \tag{restricted policies are non-responsive, thus we have a contradiction to $p$ being a policy equilibrium}
\end{align*}
where $\cP$ and $\cP^{r}$ are the set of responsive and restricted policies, respectively. 
\end{proof}

Putting all of these lemmas together yields the main structural result of this section:

\begin{theorem} \label{thm:general-eq}
For every Principal selection mechanism, there exists a pure, non-responsive Nash equilibrium of the Agent selection game (Game \ref{setting:Agent-selection}).
\end{theorem}
\begin{proof}
Follows directly from combining Lemmas~\ref{lem:pure} and~\ref{lem:extend}.
\end{proof}

\section{Monotone Selection Mechanisms And Counterfactual Regret}
\label{sec:external}
Having established in Section \ref{sec:Agent-selection} that the Agent Selection Game (Game \ref{setting:Agent-selection}) always admits a pure strategy non-responsive Nash equilibrium, we now go on to study the guarantees that the Principal can obtain in such equilibria. In particular, we will compare the utility obtained by the Principal in the dynamic Agent selection game to the utility that they could have obtained counterfactually had they offered a fixed contract to a single Agent $i$, guaranteed across all $T$ rounds.  
In evaluating the utility of the Principal under both their selection mechanism $f$ and the counterfactual fixed-agent mechanism, we assume that all agents are playing a non-responsive equilibria in the game formed by the respective selection mechanism.

The main result we show in this section is that if the Principal uses a selection mechanism that guarantees diminishing external regret (to the realized rewards) and is in addition \emph{monotone}, in a sense that we will define, then the Principal in fact obtains utility at least as high as they would have in the counterfactual benchmark, for all Agents $i$. In other words, we show that in every non-responsive equilibrium of the Agent selection game, monotone selection mechanisms suffice to convert external regret bounds into \emph{policy regret} bounds. 


\paragraph{Benchmark Comparison}
In the benchmark setting, since the Agent is promised a guaranteed contract across all $T$ rounds, they can always act so as to maximize their per-round payoff, without needing to worry about the effect that their actions today will have on whether they are chosen tomorrow; that is, they will always play the action that is myopically optimal. This is formalized via the definitions below. 

\begin{definition}[Constant Principal Selection Mechanism $f_{i}$] 
A constant Principal selection mechanism $f_{i}: \Pi \to [k] \cup \emptyset$ is defined as the mechanism which deterministically selects a fixed agent $i$ each round, regardless of the transcript.
\end{definition}

\begin{definition}[Myopic Optimal Action]
    An Agent $i$'s myopic optimal action in some round $t$ given transcript $\pi^{1:t-1}$ and corresponding previous states of nature $\bar{y}^1, \ldots, \bar{y}^{t-1}$ is an action $a^t_i$ satisfying $a^t_i \in \argmax_{a \in \cA} \mathbb{E}_{y^t \sim B^{1, t}_{i}(\bar{y}^{1:t-1})} \left[ u_i(a, y^{t}) \right]$ 
    Note that for an Agent $i$ this action may differ between rounds and may not be unique within a fixed round.
    Let the set of all such actions in round $t$ be denoted $\tilde{\cA}^t_i$.
\end{definition}

\begin{definition}[Myopic Optimal Policy]
An Agent's myopic optimal policy is defined to be a policy in which each action $p^t_i(\pi^{1:t-1})$ is a myopic optimal action for all $\pi \in \Pi$ and $t \in [T]$. The set of all myopic optimal policies for player $i$ is defined to be $p_{i}^m$.
\end{definition}

\begin{lemma}\label{lem:myopic-policy}
In every equilibrium of the game in which the Principal plays a constant Principal selection mechanism selecting agent $i$, agent $i$ plays a myopic optimal policy.
\end{lemma}
The proof of Lemma \ref{lem:myopic-policy} is deferred to Appendix \ref{app:lem-agent-incentive}.

We now define a monotone selection mechanism. Informally, a monotone selection mechanism is one in which if the Principal's reward is increased on a round in which Agent $i$ was chosen, \emph{and no other changes are made to the transcript}, then this change only increases the probability that Agent $i$ is chosen at future rounds.
\begin{definition}[Monotone Selection Mechanism] \label{def:monotone-alg}
Let $f$ be a Principal selection mechanism (Definition \ref{def:Principal-selection-mech}). 
Let $\pi^p \in \Pi^p$ be an arbitrary Principal transcript. Let $\tilde{\pi}^p$ be another Principal transcript that is identical to $\pi^p$ except that it differs in exactly a single entry $t$ such that $\tilde{\pi}^{p, t} = (i, \tilde{r}(o^t_i))$ and $\pi^{p, t} = (i, r(o^t_i))$ where $\tilde{r}(o^t_{i}) \geq r(o^t_i)$. 
The Principal's selection mechanism $f$ is said to be monotone if for all $t' \geq t$, it is the case that
$\Pr_{\cR} [ f^{t'+1}(\tilde{\pi}^{p, {1:t'}}, \cR) = i ] \geq \Pr_{\cR} [f^{t'+1}({\pi}^{p, {1:t'}}, \cR) = i ]$ --- that is, the selection mechanism is only more likely to select $i$ in future rounds under $\tilde{\pi}^{p, t}$ than $\pi^{p, t}$. 
\end{definition}

We first establish that when the Principal employs a monotone selection mechanism, then in any non-responsive pure strategy equilibrium of the game, Agents all exert effort levels that are only higher than they would have in their myopically optimal strategies. Informally, this is because we can decompose Agent utilities into their single-round utility and their continuation utility. Any effort level that is below the myopically optimal level by definition will only improve single-round utility if raised --- and because of monotonicity, will also only improve the continuation utility. By the non-responsiveness of the equilibrium, it will not affect the behavior of other Agents, and hence we can conclude that acting at any effort level below the myopically optimal level is dominated, and hence cannot be supported in equilibrium.

\begin{lemma} \label{lem:Agent-incentive}
If the Principal is selecting Agents via a monotone algorithm, in any non-responsive, pure strategy Nash equilibrium of the selection game (Game \ref{setting:Agent-selection}), then the selected Agent $i^t$'s effort levels in each round $t$ are at least as high as $\tilde{p}^t_{i^t} = \inf{\tilde{\cA}^t_{i^t}}$. 
\end{lemma}

The proof of Lemma \ref{lem:Agent-incentive} is deferred to Appendix \ref{app:lem-agent-incentive}.

With this result in hand, we are almost ready to state and prove our main result, which is that if the Principal employs a bandit learning algorithm that is both monotone and has a diminishing external regret guarantee, then if agent policies always form a non-responsive equilibrium, the Principal will have no regret to the best counterfactual world induced by picking a fixed agent across all rounds. To state the result, we first need to give some online learning preliminaries.

\subsection{Online Learning Preliminaries}
\label{sec:online-prelims}
When the Principal selects an Agent each round, their utility is determined by an unknown sequence of states of nature  and an unknown set of Agent policies. This can be embedded into an ``online learning from experts'' setting, in which the Principal selects one of the $k$ actions at each step and obtains some return based upon this selection. As the Principal does not receive full feedback and only has access to the return of the Agent selected each round, our setting is in fact a special case of  online learning with bandit feedback --- recalling of course that we ultimately want stronger than typical regret bounds. In this section we will make use of algorithms with general no-regret properties to derive no-regret guarantees for our setting. 

We begin with a brief overview of the general online learning from experts setting with bandit feedback and then introduce some definitions, including of notions of external regret and swap regret.
In this section, to align with standard notation for no-regret algorithms, we refer to losses instead of rewards. 
At every round, an adaptive adversary will choose a vector of $k$ losses --- one for each action. Simultaneously, the learning algorithm will choose a distribution over $k$ actions, and will experience loss equal to the loss of their chosen action. The learning algorithm will observe the loss of the action chosen, but will not observe the loss of actions not chosen. 
Since the adversary is allowed to be adaptive, they are able to react to the history of the game when selecting their loss vector in the current round, including to the realizations of randomness used previously by the learning algorithm (though not the algorithm's randomness in the current round). 
We make this dependence explicit in our definitions, as it will be important to our later analysis---in our applications, the loss sequence will be explicitly adaptive.

\begin{definition}[Transcript $\phi$]
    Let $\phi^{1:t}$ denote a transcript of length $t$ consisting of the history of the losses observed by the algorithm $b^1, \ldots, b^t$, the actions selected of the algorithm $s^1, \ldots, s^t$, and the randomness used by the selection algorithm $\cR^1, \ldots, \cR^t$. 
    Note that since we are in a setting with bandit feedback, the transcript does not include the entire loss vector $\ell$ but only the loss of the selected action in that round, i.e. $b^t \coloneqq \ell(\phi^{1:t-1})[s^t]$.
\end{definition}

\begin{definition}[Adaptive Adversary $\ell$]
    An adaptive adversary plays a loss function $\ell^t: \Phi^{t-1} \to \mathbb{R}^k$ in round $t$ as a function of the transcript $\phi^{1:t-1}$. 
    We refer the adaptive adversary across all rounds as $\ell$.
\end{definition}

\begin{definition}[Learning Algorithm]
    A learning algorithm $f$ is a mapping from $t-1$-length transcript prefixes $\phi^{1:t-1}_{\cR^{1:t-1}}$ and random bits used by the algorithm in round $t$, denoted $\cR^{t}$, to an action selection in the current round $t$.
\end{definition}

\begin{definition}[External Regret]
    The external regret of an algorithm $f$ against an adaptive adversary $\ell$ over $T$ rounds is defined as
        \begin{align*}
            \textsc{Reg}_{f}(T,\ell) = \max_{i \in [k]} \E_{\cR^{1:T}} \left[  \sum_{t=1}^T \ell(\phi^{1:t-1}_{\cR^{1:t-1}})[f(\phi^{1:t-1}_{\cR^{1:t-1}}, \cR^t)] -  \ell(\phi^{1:t-1}_{\cR^{1:t-1}})[i] \right], 
        \end{align*}
    where we use $\phi^{1:t-1}_{\cR^{1:t-1}}$ to denote the transcript prefix deterministically resulting from the algorithm $f$, adversary $\ell$ and a particular realization of $\cR^{1:t-1}$.
\end{definition}

\begin{definition}[Swap Function]
    Let the mapping $\chi: [k] \to [k]$ represent a strategy modification function. 
\end{definition}

\begin{definition}[Swap Regret]
    The swap regret of an algorithm $f$ against an adaptive adversary $\ell$ over $T$ rounds is defined as
    \begin{align*}
        \textsc{SReg}_{f}(T, \ell) = \max_{\chi} \E_{\cR^{1:T}} \left[  \sum_{t=1}^T \ell(\phi^{1:t-1}_{\cR^{1:t-1}})[f(\phi^{1:t-1}_{\cR^{1:t-1}}, \cR^t) ] -   \ell(\phi^{t-1}_{\cR^{1:t-1}}[ \chi(f(\phi^{1:t-1}_{\cR^{1:t-1}}, \cR^t))] \right]
    \end{align*}
      where we use $\phi^{1:t-1}_{\cR^{1:t-1}}$ to denote the transcript prefix deterministically resulting from the algorithm $f$, adversary $\ell$ and a particular realization of $\cR^{1:t-1}$.
\end{definition}

Note that in both of these definitions of regret, the benchmark comparison is to the realized loss vectors, not the counterfactual loss vectors that would have been realized had the learning algorithm made different choices. 
However, we will be able to show an even stronger guarantee in our setting - that in fact, the learning algorithm satisfies regret to these counterfactual loss vectors that they did not see, but which would have been realized under different choices of the algorithm.

\subsection{Counterfactual Guarantees in Equilibrium}
We are now in a position to state the main results of this section. We first observe that it immediately follows that if the Principal is selecting Agents using a bandit algorithm that has worst-case external regret guarantees, then they obtain utility at least as high as they would have had they consistently selected the best Agent in hindsight \emph{fixing the actions that the Agent played in Game \ref{setting:Agent-selection}.}

\begin{lemma} \label{thm:regret-realized}
If the Principal is selecting Agents with a selection mechanism with external regret bounded by $R(T)$ for all adaptive adversaries $\ell$, then in any non-responsive, pure strategy equilibrium of the game induced by Game \ref{setting:Agent-selection}, the Principal has regret at most $R(T)$ to the sequence of returns realized by the best fixed agent in hindsight. 
\end{lemma} 
\begin{proof}
    This follows immediately from instantiating the regret bound for the Principal in the Agent selection game. Note that the regret bound holds for an arbitrary adaptive adversary, and so in fact holds for any adversary with control over the states of nature in the Agent selection game. 
\end{proof}


Now, we state stronger regret guarantees. In order to state our counterfactual guarantees, we must introduce a definition for what sorts of Principal transcripts can arise given any non-responsive equilibrium of the agents.

\begin{definition}[Equilibrium Policy Set $p^{e}(f,y)$] 
Given an agent selection mechanism $f$ and a state of nature sequence $y$, $p^{e}(f,y)$ is the set of all collections of policies that Agents can play which jointly form a non-responsive equilibrium. 
\end{definition}


Now we can formally define our notion of policy regret. Informally, we compare the utility of the Principal given the mechanism they actually implemented, to the utility that the Principal would have obtained had they implemented the best fixed-agent mechanism $f_i$, given the actual sequence of realized states of nature and the beliefs of the Agents. When evaluating the utility for the Principal in both cases, we assume that the Agents play according to a non-responsive pure-strategy Nash equilibrium. Because there may be multiple such equilibria, we take the minimum Principal utility over this set. 


\begin{definition}[Policy Regret]
    For some Principal selection mechanism $f$, agent beliefs $B$, and state of nature sequence $y^{1:T}$, the policy regret of the Principal is 
    \[ \max_{i \in [k]}\min_{p \in p^{e}(f_{i},y)}\cU_p(g_{f_{i},p}(y^{1:T})) - \min_{p \in p^{e}(f,y)} \cU_p(g_{f,p}(y^{1:T})). \] 
\end{definition}


Finally,  we can state our main theorem: informally, if the Principal's algorithm has an external regret guarantee and is also monotone, then not only do they obtain no regret to the realized actions of the Agents in Game \ref{setting:Agent-selection}, but in any deterministic, non-responsive equilibrium of the induced game, they obtain utility as high as they would have had they interacted with the best of the Agents every round.

\begin{theorem} \label{thm:regret-counterfactual}
If the Principal is selecting Agents with a \emph{monotone} selection mechanism with external regret bounded by $R(T)$, the Principal has policy regret at most $R(T)$. 
\end{theorem} 
\begin{proof}
    This follows from Lemma \ref{lem:Agent-incentive}, Assumption \ref{assumption:monotone-effort-return}, and Lemma \ref{thm:regret-realized}. Any Agent's effort in the Agent selection game is only higher than their myopically optimal effort level if they are guaranteed to be selected, by Lemma \ref{lem:Agent-incentive}. 
    Next, by Assumption \ref{assumption:monotone-effort-return}, the counterfactual sequence of expected returns in the game when they were guaranteed the contract would have been only lower than the expected returns in the Agent selection game. Taken with Lemma \ref{thm:regret-realized}, that the Principal has bounded regret to the sequence of returns of the best fixed Agent in hindsight in the Agent selection game, we have the claim.
\end{proof}

To instantiate this guarantee with a concrete mechanism, we only need that a monotone bandit learning algorithm of this sort exists with a diminishing regret guarantee. Indeed, we show in Section~\ref{sec:swap} that MonoBandit(MW) is a monotone bandit selection mechanism with external regret bounded by $2\sqrt{k} \cdot \log^{\frac{1}{4}}(k) \cdot T^{\frac{3}{4}}$ (Theorem \ref{thm:algbsmw}).

\begin{theorem} \label{thm:regret-monobandit-mw}
If the Principal selects agents each round via the MonoBandit algorithm instantiated with Multiplicative Weights, the Principal has policy regret at most $2\sqrt{k} \cdot \log^{\frac{1}{4}}(k) \cdot T^{\frac{3}{4}}$.
\end{theorem}
\begin{proof}
    This follows from Theorem \ref{thm:regret-counterfactual} and Theorem \ref{thm:algbsmw}.
\end{proof}




\section{Swap Regret and Limited Liability Contracts}
\label{sec:swap}

In Section \ref{sec:external}, we showed that by using a monotone selection mechanism with external regret guarantees, the Principal is able to guarantee payoff similar to what they could have obtained by using the same contract in a benchmark setting in which the contract was guaranteed to any single fixed Agent across all $T$ rounds. However, if we take the benchmark contract to be one that is not \emph{limited liability} --- i.e. that can sometimes obligate the Agent to pay the Principal --- then our mechanism also must offer contracts which fail to satisfy the limited liability condition. This is problematic, because on the one hand, the natural, robust set of linear benchmark policies do not satisfy the limited liability condition, but on the other hand, contracts that fail to satisfy the limited liability condition can be difficult to implement in practice. In this section, we turn our attention to mechanisms for the Principal that only offer limited liability contracts, while still being able to offer the same kinds of counterfactual regret guarantees we derived in Section \ref{sec:external} with respect to linear contracts, which are not limited liability.

\begin{definition}[Limited Liability Contract]
    A limited liability contract is a payment rule $v: \mathbb{R} \to \mathbb{R}_{\geq 0}$ mapping the Principal's return to a payment made to an Agent.
\end{definition}

\begin{definition}[Linear Contract]
    A linear contract $v_{\alpha}: \mathbb{R} \to \mathbb{R}$ with parameter $\alpha \in [0, 1]$ is a payment rule mapping the Principal's return to a reward to the Agent as an $\alpha-$fraction of the total return $r$, as $v_{\alpha}(r) = \alpha \cdot r$. 
\end{definition}

To start, we prove that, by offering a linear contract each round and playing a monotone no swap-regret selection mechanism, the Principal guarantees that the liability of each agent is bounded by the swap regret term. Then, we prove that, if we equip this selection mechanism with a tab that ensures that the payment to each agent fulfills limited liability, it is an approximate (up to regret terms) equilibrium for the Agents to play a strategy profile corresponding to a non-responsive equilibrium of the general game. We use these results to show that, if we assume that agents will play one of these approximate non-responsive equilibria, the Principal gets low policy regret while ensuring limited liability for the Agents.  

\begin{lemma}
If the Principal selection mechanism $f$ has swap regret bounded by $\textsc{SReg}(T)$ against any adaptive adversary $\ell$, and the Principal is providing a fixed linear contract $v_\alpha$ each round, then each agent's cumulative payment is at least $- \alpha \textsc{SReg}(T)$.  \label{lem:liability_bound}
\end{lemma}
\begin{proof} Let $\mathcal{U}_{p}(f,P,y|i)$ represent the Principal utility from the rounds in which they select agent $i$. The Principal selection mechanism $f$ has swap regret at most $\textsc{SReg}_f(T, \ell)$ for any adaptive adversary $\ell$. Furthermore, the Principal always has an outside option, denoted $\emptyset$, that guarantees them return $0$ in every round. Therefore, we have that 
\begin{align*}
\mathcal{U}_{p}(f,P,y) &= 
       \mathbb{E}_{\pi^{1:T} \sim g^{1:T}_{f,p}(y^{1:T})} \left[\sum_{t' = 1}^{T} \mathbb{I}[(f^t(\pi^{p,1:t'-1}, \cR^t) = i)] \cdot u_p(i,p_{i}(\pi^{1:t'-1}), y^{t'}) \right] \\
       &\geq \mathcal{U}_{p}(\chi(f),P,y) - \textsc{SReg}(T) \\
       & = \mathbb{E}_{\pi^{1:T} \sim g^{1:T}_{f,p}(y^{1:T})} \left[\sum_{t' = 1}^{T} \mathbb{I}[(\chi(f^t(\pi^{p,1:t'-1}, \cR^t) = i))] \cdot u_p(i,p_{i}(\pi^{1:t'-1}), y^{t'}) \right] - \textsc{SReg}(T),
\end{align*}
for any swap function $\chi$, and in particular, the swap function $\chi(j) = \begin{cases}
    j, \text{if } j \neq i, \\
    \emptyset, \text{if } j = i
\end{cases}$.

So we have that 
\begin{align*}
\mathcal{U}_{p}(f,P,y)&  \geq \mathbb{E}_{\pi^{1:T} \sim g^{1:T}_{f,p}(y^{1:T})} \left[\sum_{t' = 1}^{T} \mathbb{I}[f^t(\pi^{p,1:t'-1}, \cR^t) = i] \cdot u_p(\emptyset,p_{i}(\pi^{1:t'-1}), y^{t'}) \right] - \textsc{SReg}(T) \\
&  \rightarrow \mathcal{U}_{p}(f,P,y|i) \geq \mathbb{E}_{\pi^{1:T} \sim g^{1:T}_{f,p}(y^{1:T})} \left[\sum_{t' = 1}^{T} \mathbb{I}[f^t(\pi^{p,1:t'-1}, \cR^t) = i] \cdot u_p(\emptyset,p_{i}(\pi^{1:t'-1}), y^{t'}) \right] - \textsc{SReg}(T) \\
&  \rightarrow \mathcal{U}_{p}(f,P,y|i) \geq - \textsc{SReg}(T) \tag{by the fact that the reward and cost of the outside option are both $0$}
\end{align*}

Thus, for all Agent $i$, the total utility the Principal gains during these rounds is at most $-\textsc{SReg}(T)$. As the Principal is providing a fixed linear contract $v_\alpha$ each round, the total payment from the Principal to Agent $i$ is $\alpha \cdot \mathcal{U}_{p}(f,P,y|i) \geq -\alpha \textsc{SReg}(T)$.
\end{proof}

We show in Section~\ref{sec:swap} that our algorithm MonoBandit when instantiated with the algorithm TreeSwap from~\cite{dagan2023external} is a monotone bandit selection mechanism with swap regret bounded by $2\sqrt{k} \cdot o(T)$ (Theorem \ref{thm:algbsts}).

\begin{theorem} \label{thm:regret-monobandit-ts}
If the Principal selects agents each round via the MonoBandit algorithm instantiated with TreeSwap, the Principal has policy regret at most $2\sqrt{k} \cdot o(T)$.
\end{theorem}
\begin{proof}
    This follows from Theorem \ref{thm:regret-counterfactual} and Theorem \ref{thm:algbsts}.
\end{proof}

\begin{theorem}
If the Principal selects agents each round via the MonoBandit algorithm (Algorithm \ref{alg:monobandit}) instantiated with TreeSwap and provides a fixed linear contract, then:\\
1) the Principal has policy regret at most $2\sqrt{k}\cdot o(T)$, and \\
2) all agents have expected liability no more than $2\sqrt{k} \cdot o(T)$.
\end{theorem}
\begin{proof}
    The first claim follows from Theorem \ref{thm:regret-monobandit-ts}. The second claim follows from Lemma \ref{lem:liability_bound}.
\end{proof}

Thus, by running a monotone no swap-regret algorithm in the bandit setting and offering a linear contract, the Principal ensures that after the interaction, the net payment from the Principal to the Agent must be non-negative (up to a diminishing regret term).

We now introduce a limited liability mechanism. Informally, it works as follows: the Principal opens a ``tab'' for each Agent that allows them to accumulate debt.
The Principal maintains a running total of the cumulative payments would be made under the linear contract, but defers the payments to a single lump sum transfer at end of the interaction. At the end of the interaction, if the tab is negative, any remaining debt owed from the Agent to the Principal is forgiven, and so no payments ever need to be made from the Agent to the Principal. Otherwise, the Principal pays the agent the tab.
Introducing this tab does not impact the net payment made from the Principal to the Agent whenever the original net payment is non-negative. However, since the tab results in a zero payment when the net payment from the Principal to the Agent would have been negative, this mechanism is incentive distorting. Informally, since we have established that the net payment from the Principal to the Agent cannot have been very negative, we can show that the distortion of incentives is minimal.

We formalize the limited liability game in Game \ref{setting:Agent-selection-limited}. It is similar to the standard game, with the exception of the addition of the three boldface lines.

\begin{game}
\caption{Agent Selection with Limited Liability Game}
\label{setting:Agent-selection-limited}
    \begin{algorithmic}
        \STATE {\bf Input} fixed linear contract $\alpha$, time horizon $T$
        \STATE Agent cost functions $c_i(\cdot)$ and belief functions $B_i(\cdot)$ are common knowledge amongst all Agents. 
        \STATE Principal fixes their selection mechanism $f$, which is common knowledge to all Agents.
        \STATE Agent transcript is initialized their transcript $\pi = \{\}$
        \STATE Principal initializes their transcript of returns $\pi^p = \{\}$
        \STATE {\bf Principal initializes a tab $\cT_i$ for each Agent $i \in [k]$}
        \FOR{$t \in [T]$}
            \STATE Principal selects Agent $f(\pi^{p, t-1}) = i^t$ to whom to award contract using random bits $\cR^t$
            \STATE Agent $i^t$ takes action $a^t$, observed by all Agents, at a cost of $c_{i^t}(a^t)$
            \STATE Agents observe the state of nature $y^t$
            \STATE Agent $i^t$ observes their outcome $o_{i^t}^t \sim \cD_{i^t, y^t, a^t}$        
            \STATE Principal receives the return from realized outcome $r(o_{i^t}^t)$
            \STATE {\bf Principal updates Agent $i^t$'s tab: $\cT_{i^t} = \cT_{i^t} + \alpha \cdot r(o_{i^t}^t)$}
            \STATE Nature reveals randomness used to sample the outcome $\cR^t_{i^t}$
             \STATE Principal publishes randomness $\cR^t$ used in round $t$ 
            \STATE Principal transcript is updated with $(i^t, r(o_{i^t}^t))$
            \STATE Agent transcript is updated with $(y^t, \cR^t_{i^t}, \cR^t, a^t_{i^t})$
        \ENDFOR
        \STATE {\bf Principal pays $\max(0, \cT_i)$ to each Agent $i \in [k]$}
    \end{algorithmic}
\end{game}

We first establish that if the Principal employs a selection mechanism that has a swap regret guarantee, that for any set of policies played by the Agents, the Agents' utilities are almost identical (up to the swap regret bound) in the Limited Liability Game (Game \ref{setting:Agent-selection-limited}) and in the original game (Game \ref{setting:Agent-selection}). Informally, this is because whenever the Agent ends the interaction without ``debt'', their payments are identical in the two settings, and the swap regret guarantee that the Principal has to the outside option (which always guarantees payoff 0) will similarly guarantee that no Agent will ever end the game with significant debt (because by the nature of linear contracts, this would also imply significant swap regret for the Principal). 

\begin{definition}[Agent Utility in the Limited Liability Game $\cU^{Lim}$] Let $\bar{v}_{i}(\pi)$ represent the total payment from the Principal to the Agent $i$ under the transcript $\pi$. Furthermore, let $c_{i}(\pi)$ represent the sum of the costs of all actions taken by Agent $i$ (on rounds they are selected). Then, let $\cU^{Lim}$ be defined as the utility of Agent $i$ in the limited liability game:
\[ \mathcal{U}^{Lim}_{i}(1,\emptyset, p_{j \neq i},p_{i}) = \sum_{\pi \in \Pi} (\max(0,\bar{v}_{i}(\pi)) - c_{i}(\pi) )\cdot \Pr[g_{f,p_{j\neq i}, p_{i}}(B^1_i) = \pi]  \] 
\end{definition}

\begin{lemma} \label{lem:limited-liability-approx}
Fix any set of Agent beliefs $B$, set of Agent policies $p$, and linear contract $v_\alpha$. If the Principal selection mechanism $f$ has swap regret bounded by $\textsc{SReg}(T)$ against any adaptive adversary $\ell$, then each Agent will have expected payoff that differs by no more than $\alpha \textsc{SReg}(T)$ between the standard Agent selection game (Game \ref{setting:Agent-selection}) and the limited liability Agent selection game (Game \ref{setting:Agent-selection-limited}).
That is, for any Agent $i \in [k]$,
    \[ \cU_i(1, \emptyset, p_{j\neq i}, p_i) + \alpha \textsc{SReg}(T) \geq \cU^{Lim}_i(1, \emptyset, p_{j\neq i}, p_i) \geq \cU_i(1, \emptyset, p_{j\neq i}, p_i). \]
\end{lemma}
\begin{proof}

We can write the utility for an Agent $i$ under the limited liability game as 
\begin{align*}
  \mathcal{U}^{Lim}_{i}(1,\emptyset, p_{j \neq i},p_{i}) &=  \sum_{\pi \in \Pi} (\max(0,\bar{v}_{i}(\pi)) - c_{i}(\pi) )\cdot \Pr[g_{f,p_{j\neq i}, p_{i}}(B^1_i) = \pi]    \\
   &= \sum_{\pi \in \Pi} (\bar{v}_{i}(\pi) + \max(0,-\bar{v}_{i}(\pi)) - c_{i}(\pi)) \cdot  \Pr[g_{f,p_{j\neq i}, p_{i}} (B^1_i) = \pi^{*}]  \\
    &= \sum_{\pi \in \Pi} (\bar{v}_{i}(\pi) - c_{i}(\pi)) \cdot  \Pr[g_{f,p_{j\neq i}, p_{i}} (B^1_i) = \pi^{*}] + \sum_{\pi \in \Pi} (\max(0,-\bar{v}_{i}(\pi))) \cdot  \Pr[g_{f,p_{j\neq i}, p_{i}} (B^1_i) = \pi^{*}]  \\
    &= \mathcal{U}_{i}(1,\emptyset, p_{j \neq i},p_{i}) + \sum_{\pi \in \Pi} (\max(0,-\bar{v}_{i}(\pi))) \cdot  \Pr[g_{f,p_{j\neq i}, p_{i}} (B^1_i) = \pi^{*}].
\end{align*}

First, observe that the right hand side of the desired inequality follows immediately from the above. This is simply because the total payment made to any Agent in the limited liability game is exactly the max of zero and the total payments made to that Agent in the non-limited liability game. Therefore, each Agent only has weakly higher utility in the limited liability game. 

Now, note that by Lemma~\ref{lem:liability_bound}, any Agent's expected cumulative payment in the non-limited liability game is bounded below by $- \alpha \textsc{SReg}(T)$. In the non-limited liability game, this would mean a transfer from the Agent to the Principal of $\alpha \textsc{SReg}(T)$. Since payments to the Agents between the limited liability and non-limited liability games only differ when the latter are negative, we have established that the utility of an agent is at most $\alpha \textsc{SReg}(T)$ higher in the limited liability game.
\end{proof}

We now establish that any equilibrium of the general Agent selection Game (Game \ref{setting:Agent-selection}) remains an approximate equilibrium of the limited liability game (Game \ref{setting:Agent-selection-limited}) whenever the Principal employs a no swap-regret algorithm. Hence, the guarantees that we showed in Section \ref{sec:external} for the Principal in equilibrium (that they obtain payoff at least as high as they would have in the counterfactual world in which the Principal engaged in a constant selection mechanism with the best Agent $i$ in hindsight) continue to hold here (in approximate equilibrium), whenever the Principal employs a \emph{monotone} selection algorithm with swap-regret guarantees.

\begin{theorem}
    Fix a linear contract $v_\alpha$ and a  Principal selection mechanism $f$ with swap regret bounded by $\textsc{SReg}(T)$ against any adaptive adversary. Any
    Nash equilibrium in the game with Agent selection (Game \ref{setting:Agent-selection}) is an $\alpha\textsc{SReg}(T)$-approximate equilibrium in the limited liability game (Game \ref{setting:Agent-selection-limited}). \label{thm:liable} 
\end{theorem}
\begin{proof}
    Let $p$ be any pure non-responsive equilibrium of the Agent selection game. We show that no Agent can unilaterally deviate to improve their expected payoff by more than $\textsc{SReg}(T)$. 
    Let $p'$ be a strategy profile that is defined as $p_{:-i}, p'_i$, i.e. is identical to $p$ except for an arbitrary deviation by Agent $i$.
    The claim follows by two applications of Lemma \ref{lem:limited-liability-approx}:
    \begin{align*}
        \cU^{Lim}_i(1, \emptyset, p_{j\neq i}, p_i) &\geq 
        \cU_i(1, \emptyset, p_{j\neq i}, p_i) \\
        \cU_i(1, \emptyset, p_{j\neq i}, p'_i) + \alpha R(T) &\geq \cU^{Lim}_i(1, \emptyset, p_{j\neq i}, p'_i).
    \end{align*} 
    As $p$ is assumed to be an exact Nash equilibrium of the general game, we also know
    \begin{align*}
        \cU_i(1, \emptyset, p_{j\neq i}, p_i) \geq  \cU_i(1, \emptyset, p_{j\neq i}, p'_i).
    \end{align*}
    Together, these give the claim:
    \begin{align*}
        \cU^{Lim}_i(1, \emptyset, p_{j\neq i}, p_i) &\geq \cU^{Lim}_i(1, \emptyset, p_{j\neq i}, p'_i) - \alpha \textsc{SReg}(T). 
    \end{align*} 
\end{proof}

Defining the policy space over which to evaluate policy regret is slightly more delicate now; previously, we assumed that all agents were playing a non-responsive equilibrium of the full liability game. However, now that we are in a limited liability setting, the incentives of agents may have shifted by regret terms. Thus, we have strong guarantees not over all non-responsive pure-strategy equilibria of this new game but over another  reasonable set of equilibria the agents could play: the set of non-responsive pure-strategy equilibria of the general, full liability game, which we have shown are all approximate equilibria of the limited liability game.

\begin{definition}[Policy Regret in the Limited-Liability Space]
    For some set of agent beliefs $B$ and state of nature sequence $y^{1:T}$, the policy regret of a Principal implementing selection mechanism $f$ in Game~\ref{setting:Agent-selection-limited} is 
  \[ \max_{i \in [k]}\min_{p_{i} \in p_{i}^{m}}\cU_p(g^{i}_{p_{i}}(y^{1:T})) - \min_{p \in p^{e}(f,y)} \cU_p(g^{l}_{f,p}(y^{1:T})) \]
\end{definition}

\begin{theorem}
When the Principal uses the MonoBandit selection mechanism (Algorithm \ref{alg:monobandit}) and employs a fixed, linear contract $v_\alpha$, for any non-responsive equilibrium $p$ of the general game:
\begin{itemize}
\item $p$ is a $2\alpha\sqrt{k} \cdot o(T)$-approximate equilibrium of the limited liability game
\item the Principal has policy regret at most $2\alpha\sqrt{k}\cdot o(T)$.
\end{itemize}
  
\end{theorem}

\begin{proof}
By Theorem~\ref{thm:algbsts}, MonoBandit(TreeSwap) is a monotone algorithm with swap regret bounded by $2\alpha\sqrt{k}\cdot o(T)$. Thus, by Theorem~\ref{thm:liable}, if all agents play according to an equilibrium of the full liability game, this is an $2\alpha\sqrt{k}\cdot o(T)$-approximate equilibrium of the limited liability game. Finally, we get the policy regret bound for the Principal by Theorem \ref{thm:regret-counterfactual}.
\end{proof}

\section{Monotone No Swap-Regret Algorithms}

Finally, it remains to establish the existence of a bandit no-regret learning algorithm that is simultaniously monotone and obtains a non-trivial swap regret guarantee against an adaptive adversary. First we note that even in the full information setting (when the learning algorithm gets to observe the entire loss vector, and not just the loss of the action they choose), standard no swap-regret algorithms fail to satisfy the monotonicity property. 
\begin{lemma} \label{lem:blummansour-not-monotone}
The Blum-Mansour algorithm (Algorithm \ref{alg-bm}) is not monotone (Definition \ref{def:monotone-alg}). 
\end{lemma}

We verify this numerically by showing two loss sequences in which Blum-Mansour violates the monotonicity condition. This example can be found in Appendix~\ref{app:blum}.

On the other hand, the very recent swap regret algorithm given by \cite{dagan2023external} in the full information setting is indeed monotone.

\begin{lemma} \label{lem:treeswap-monotone}
The TreeSwap algorithm (Algorithm 1 of \cite{dagan2023external}), when instantiated with a monotone no-external regret algorithm, is monotone (Definition \ref{def:monotone-alg}). 
\end{lemma}
The proof of this claim is deferred to Appendix \ref{app:treeswap-monotone}.

\subsection{A Monotone Bandit No Swap-Regret Algorithm in the Adaptive Setting}

For our application, however, we need a no swap-regret algorithm that is monotone and operates in the bandit setting. In this section we give a monotonicty preserving reduction from the full-information setting to the bandit learning setting, such that when paired with a monotone swap-regret algorithm in the full information setting (like that of \cite{dagan2023external}) we obtain the object we need.

The MonoBandit algorithm preserves any combination of the following appealing properties from a full-feedback learning algorithm: monotonicity, sublinear swap regret, and sublinear external regret. These proofs are all deferred to Appendix \ref{app:monomonomono}.

\begin{lemma}
If $\textsc{SReg}_{\cA}(T,\ell) \leq R(T)$ $\forall \ell$, then $\textsc{SReg}_{MonoBandit(\cA)}(T,\ell) \leq 2\sqrt{k \cdot T \cdot R(T)}$, $\forall \ell$. \label{thm:sr-reduction}
\end{lemma}

\begin{lemma}
If $\textsc{Reg}_{\cA}(T,\ell) \leq R(T)$ $\forall \ell$, then $\textsc{Reg}_{MonoBandit(\cA)}(T,\ell) \leq 2\sqrt{k \cdot T \cdot R(T)}$, $\forall \ell$. \label{thm:r-reduction}
\end{lemma}

\begin{lemma} \label{thm:mono-reduction}
If algorithm $\cA$ is monotone, then MonoBandit($\cA$) is monotone.
\end{lemma}

\begin{theorem}
$MonoBandit(MW)$ is a monotone algorithm and 
$$\textsc{Reg}_{MonoBandit(MW)}(T) = 2\sqrt{k \cdot \sqrt{log(k)}} \cdot T^{\frac{3}{4}}.$$ \label{thm:algbsmw}
\end{theorem}

\begin{proof}
As Multiplicative Weights is a monotone algorithm with $\textsc{Reg}_{MW}(T) = \sqrt{\log(k) \cdot T}$, by Lemmas~\ref{thm:r-reduction} and~\ref{thm:mono-reduction}, MonoBandit(MW) is monotone and 
 \begin{align*}
 \textsc{Reg}_{MonoBandit(MW)}(T) & \leq 2\sqrt{T \cdot k \cdot \sqrt{\log(k) \cdot T}} \\
 & = 2\sqrt{k \cdot \sqrt{\log(k)}} \cdot T^{\frac{3}{4}}.
 \end{align*}
\end{proof}

\begin{theorem} \label{thm:algbsts}
$MonoBandit(TreeSwap)$ is a monotone algorithm and 
$$\textsc{SReg}_{MonoBandit(TreeSwap)}(T) = 2\sqrt{k} \cdot o(T).$$
\end{theorem}

\begin{proof}

From~\cite{dagan2023external}, the algorithm TreeSwap has the property that $\textsc{SReg}_{TreeSwap}(T) = o(T)$.

 Combining this with Lemma~\ref{lem:treeswap-monotone}, we get that TreeSwap is a monotone algorithm with $\textsc{SReg}_{TreeSwap}(T) = o(T)$. Therefore, by Lemmas~\ref{thm:sr-reduction} and~\ref{thm:mono-reduction}, MonoBandit(TreeSwap) is monotone and 
 \begin{align*}
 \textsc{SReg}_{MonoBandit(TreeSwap)}(T) & \leq 2\sqrt{T \cdot k \cdot o(T)} \\
 & = 2 \sqrt{k} \cdot \sqrt{o(T^{2})} \\
 & = 2 \sqrt{k} \cdot o(T).
 \end{align*}
\end{proof}

\section{Conclusion and Discussion}
In this paper we have studied the utility of a Principal in equilibrium in a complex $k$-player, $T$-stage extensive form Agent selection game, and have given mechanisms that ensure that in equilibrium, the Principal has no regret to the \emph{counterfactual} world in which they had offered a guaranteed contract to the best fixed Agent in hindsight, even taking into account the Agent's differing behavior in this world. The primary tools that we have used to derive guarantees in this complex setting are online learning algorithms and their properties. In particular, we have connected  external regret guarantees for bandit learning algorithms paired with \emph{monotonicity} with counterfactual policy regret guarantees in these games, and have connected swap regret guarantees with the ability to preserve the incentive properties of linear contracts in limited liability implementations. Moreover, we have established the existence of bandit learning algorithms that are simultaneously monotone and have swap regret guarantees, which instantiates our theorems with concrete mechanisms. 

An objection that one can level at our limited liability mechanism is that it defers all payments until the final round. As such, the Principal can incur substantial debt to the Agent. It would be preferable if the Principal immediately paid any debt owed to the Agent, and still need not ever collect debt owed from the Agent to the Principle. We note that in this implementation, the mechanism would still induce similar incentives to that of the full liability mechanism if the Principal's selection rule had no \emph{adaptive} swap regret --- i.e. no swap regret not just overall (as summed from rounds $1$ to $T$), but also on every  subsequence  $[t,T]$ simultaneously for all $t < T$. Algorithms with no adaptive swap regret exist --- see e.g. \cite{lee2022online} for a generic recipe for deriving such algorithms --- but we are unaware of any such algorithms that are simultaneously \emph{monotone}, which is a necessary ingredient for our policy regret guarantees. This is a concrete setting where a new result in online learning (the existence of a monotone bandit algorithm obtaining adaptive swap regret guarantees) would have an immediate application in mechanism design.

\bibliographystyle{plainnat}
\bibliography{bib}

\appendix 

\section{Proofs from Section~\ref{sec:existence}}

\subsection{Proof of Lemma~\ref{lem:subgame-equivalence}}
\label{app:subgame-equivalence}
\begin{proof}[Proof of Lemma \ref{lem:subgame-equivalence}]
\begin{align*}
    \mathcal{U}_{i}(1,\emptyset, p_{j \neq i},\hat{p}_{i}) &= \E_{\pi^{1:T} \sim g^{r, 1:T}_{f}(B^1_i)} \left[ \sum_{t' = t}^{T} \mathbb{I}[(f^t(\pi^{p, 1:t'-1}, \cR^t) = i)] \cdot u_i(p^{t'}_{i}(\pi^{1:t'-1}), y^{t'}) \right] \\
    &= \sum_{\pi^* \in \Pi^{r, 1:t-1}} \Pr [g^{r}_{f}(B^1_i) = \pi^{*}] \cdot (\mathcal{U}(\pi^{*}) + \mathcal{U}_{i}(t, \pi^{*}, p_{j \neq i},\hat{p}_{i})) \\
    &= \sum_{\pi^* \in \Pi^{r, 1:t-1}}  \Pr[g^r_{f}(B^1_i) = \pi^{*}]  \cdot (\mathcal{U}(\pi^{*}) + \mathcal{U}_{i}(t, \pi^{*}, p_{j \neq i},\hat{p}_{i})) \\ \tag{as all rounds before round $t$ are independent of this change} \\
    &= \sum_{\pi^* \neq \pi^{r,1:t-1} \in \Pi^{r, 1:t-1}}  \Pr [g^r_{f}(B^1_i) = \pi^{*}] \cdot (\mathcal{U}(\pi^{*}) + \mathcal{U}_{i}(t, \pi^{*}, p_{j \neq i},\hat{p}_{i}))
\end{align*}
\begin{align*}
    & \quad \quad +  \Pr[g^{r, 1:t-1}_{f} (B^1_i) = \pi^{r,1:t-1}] \cdot (\mathcal{U}(\pi^{r,1:t-1}) + \mathcal{U}_{i}(t, \pi^{r,1:t-1}, p_{j \neq i},\hat{p}_{i}))  \\
    &= \sum_{\pi^* \neq \pi^{r,1:t-1} \in \Pi^{r, 1:t-1}} \Pr[g^r_{f} (B^1_i) = \pi^{*}] \cdot (\mathcal{U}(\pi^{*}) + \mathcal{U}_{i}(t, \pi^{*}, p_{j \neq i},p_{i})) \\ 
    & \quad \quad + \Pr[g^{r, 1:t-1}_{f} (B^1_i) = \pi^{r,1:t-1}] \cdot (\mathcal{U}(\pi^{r,1:t-1}) + \mathcal{U}_{i}(t, \pi^{r,1:t-1}, p_{j \neq i},\hat{p}_{i})) \tag{as $\hat{p}_{i}$ is equivalent to $p_{i}$ under any transcript prefixes that are not $\pi^{r,1:t-1}$} 
\end{align*}
Therefore,
\begin{align*}
    \mathcal{U}_{i}(1,\emptyset, p_{j \neq i},\hat{p}_{i}) &- \mathcal{U}_{i}(1,\emptyset, p_{j \neq i},p_{i}) \\
     &= \sum_{\pi^* \neq \pi^{r,1:t-1} \in \Pi^{r, 1:t-1}} \Pr[g^r_{f} (B^1_i) = \pi^{*}] \cdot \left(\mathcal{U}(\pi^{*}) + \mathcal{U}_{i}(t, \pi^{*}, p_{j \neq i},p_{i}) \right) \\ 
     & \quad \quad \quad +  \Pr[g^{r, 1:t-1}_{f} (B^1_i) = \pi^{r,1:t-1}]  \cdot \left(\mathcal{U}(\pi^{r,1:t-1}) + \mathcal{U}_{i}(t, \pi^{r,1:t-1}, p_{j \neq i},\hat{p}_{i}) \right) \\
     & \quad \quad \quad - \sum_{\pi^* \neq \pi^{r,1:t-1} \in \Pi^{r, 1:t-1}} \Pr[g^r_{f} (B^1_i) = \pi^{*}] \cdot \left(\mathcal{U}(\pi^{*}) + \mathcal{U}_{i}(t, \pi^{*}, p_{j \neq i},p_{i}) \right) \\ 
     & \quad \quad \quad - \Pr[g^{r, 1:t-1}_{f} (B^1_i) = \pi^{r,1:t-1}] \cdot \left(\mathcal{U}(\pi^{r,1:t-1}) + \mathcal{U}_{i}(t, \pi^{r,1:t-1}, p_{j \neq i},p_{i}) \right) \\
     &= \Pr[g^{r, 1:t-1}_{f} (B^1_i) = \pi^{r,1:t-1}] \cdot \left(\mathcal{U}_{i}(t, \pi^{r,1:t-1}, p_{j \neq i},\hat{p}_{i}) - \mathcal{U}_{i}(t, \pi^{r,1:t-1}, p_{j \neq i},p_{i}) \right). 
\end{align*}
\end{proof}

\subsection{Proof of Lemma~\ref{lem:pure}} \label{app:pure}
\begin{proof}[Proof of Lemma \ref{lem:pure}]


By Theorem \ref{thm:glicksberg}, we know that the restricted game has a mixed strategy Nash equilibrium. Let us call this equilibrium strategy profile $p^{*}$.
We will now show that there is also an exact pure strategy Nash equilibrium in the restricted game. To do this, we will alter $p^*$ until it is pure, while still ensuring that it is a Nash equilibrium. First, we will consider some transcript prefix $\bar{\pi}^{r,1:t-1}$ with states of nature $\bar{y}^{1:t-1}$ in which player $i$'s policy is outputting a distribution, and we will modify their policy to $\hat{p}$ so that they instead deterministically output the expectation of their effort levels in $p^{*}_{i}(\pi^{r,1:t-1})$. By Lemma~\ref{lem:subgame-equivalence}, this change will be confined only to the Agent utility at that subgame. 

Recall that an Agent $i$'s expected payoff of playing  $p_i$ is as follows:
\begin{align*}
    &\mathbb{E}_{\pi^{r, t:T} \sim g^r_{f, p}(\bar{\pi}^{r, 1:t-1},B^{t, t:T}_i (\bar{y}^{1:t-1}))} \left[\sum_{t' = t}^{T} \mathbb{I}[f(\pi^{p, 1:t'-1}, \cR^t) = i] \cdot u_i(p^{t'}_{i}(\pi^{r, 1:t'-1}), y^{t'}) \right] \\
    & = \mathbb{E}_{\pi^{r, 1:t} \sim g^{r, 1:t}_{f, p}(\bar{\pi}^{r, 1:t-1},B^{t, t}_i (\bar{y}^{1:t-1}))} \left[ f(\pi^{p, 1:t-1}, \cR^t) = i] \cdot u_i(p^{t}_{i}(\pi^{r, 1:t-1}), y^{t}) \right] \\ 
    & \quad + \mathbb{E}_{\pi^{r, t:T} \sim g^r_{f, p}(\bar{\pi}^{r, 1:t-1},B^{t, t:T}_i (\bar{y}^{1:t-1}))} \left[\sum_{t' = t+1}^{T} \mathbb{I}[f(\pi^{p, 1:t'-1}, \cR^t) = i] \cdot u_i(p^{t'}_{i}(\pi^{r, 1:t'-1}), y^{t'}) \right].
\end{align*}


We will first prove that if Agent $i$ plays a pure strategy given by the expectation of the mixed strategy $p^{*, t}_{i}$ against $p^{*, t}_{-i}$, they will receive only higher utility in expectation. Note that player $i$ changing their actions will not affect the actions of any other players under $p^{*}_{-i}$, as all players are playing restricted policies. Thus, we can consider the actions of the other $k-1$ Agents as fixed when considering the payoff change for Agent $i$.
We consider the two parts of player $i$'s utility separately. 

First, we will consider their difference in their immediate payoff in round $t$ when switching from their original mixed strategy $p$ to the deterministic strategy $\hat{p}$. Note that before round $t$, the policies $p_i$ and $\hat{p}_i$ are identical. 

\begin{align*}
    &\mathbb{E}_{\pi^{1:t} \sim g^{r, 1:t}_{f}(\bar{\pi}^{1:t-1}, B^{t, t}_i(\bar{\pi}^{1:t-1}))} \mathbb{I} \left[ f(\pi^{p, 1:t-1}, \cR^t) = i] \cdot u_i(\hat{p}^{t}_{i}(\pi^{r, 1:t-1}), y^t) \right] \\ 
    & - \mathbb{E}_{\pi^{r, 1:t} \sim g^{r, 1:t}_{f}(\bar{\pi}^{1:t-1}, B^{t, t}_i(\bar{y}^{1:t-1}))} \mathbb{I} \left[ f(\pi^{p, 1:t-1}, \cR^t) = i] \cdot u_i(p^{t}_{i}(\pi^{r, 1:t-1}), y^{t}) \right] \\
    &= \mathbb{E}_{\pi^{r, 1:t} \sim g^{r, 1:t}_{f}(\bar{\pi}^{1:t-1}, B^{t, t}_i(\bar{y}^{1:t-1}))} \mathbb{I} \left[ f(\pi^{p, 1:t-1}, \cR^t) = i ] \cdot (u_i(\hat{p}^{t}_{i}(\pi^{r, 1:t-1}), y^{t}) - u_i(p^{t}_{i}(\pi^{r, 1:t-1}), y_{t})) \right].
\end{align*}

Recall that the probability of an outcome $o$ for Agent $i$ given action $a$ and state of nature $y$ is $p_{i, o} (a, y) = m_{i, o, y} \cdot a + b_{i, o, y}$.
We can now write the difference in the expected single round utility for Agent $i$ when deviating from $p^t_i$ to $\hat{p}^t_i$.

\begin{align*}
& \E_{\hat{a} \sim \hat{p}^t_i(\pi^{1:t-1})} [u_i(\hat{a}, y^t) ]
- \E_{a \sim p^{t}_i(\pi^{1:t-1}) }  [ u_i(a, y^t)] \\
& = \E_{\hat{a} \sim \hat{p}^t_i(\pi^{1:t-1})} [\bar{v}(i, \hat{a}, y^t) - c_i(\hat{a})] - 
 \E_{a \sim p^{t}_i(\pi^{1:t-1}) } [\bar{v}(i, a, y^t) - c_i(a)]. 
\end{align*}

By the concavity of $v(\cdot)$ and the convexity of $c_i(\cdot)$, this value is always non-negative. Thus, the expected immediate payoff change for Agent $i$ from moving from action $p^t_i$ to $\hat{p}^t_i$ in round $t$ is non-negative. 

Next, consider the change in the expected continuation payoff at round $t$. Note that for rounds $t' > t$, the strategies $p^{t'}_i$ and $\hat{p}^{t'}_i$ are identical once again.

\begin{align*}
   & \mathbb{E}_{\pi^{r, 1:T} \sim g^{r, 1:T}_{f}(\bar{\pi}^{r, 1:t-1}, B^{t, t:T}_i(\bar{y}^{1:t-1})))} \left[\sum_{t' = t+1}^{T} \mathbb{I}[f(\pi^{p, 1:t'-1}, \cR^t) = i] \cdot u_i(p^{t'}_{i}(\pi^{r, 1:t'-1}), y^{t'}) \right]  \\ 
    & - \mathbb{E}_{\pi^{t:T} \sim g^{r, 1:T}_{f}(\bar{\pi}^{r, 1:t-1}, B^{t}_i(\bar{y}^{1:t-1})} \left[\sum_{t' = t+1}^{T} \mathbb{I}[f(\pi^{p, 1:t'-1}, \cR^t) = i] \cdot u_i(p^{t'}_{i}(\pi^{r, 1:t'-1}), y^{t'}) \right].
\end{align*}

By Lemma~\ref{lem:p_tran_dist}, the distributions over Principal transcripts in both expressions are equal. Since Agents are playing restricted policies, no future Agent actions will change if Agent $i$ moves from strategy $p_i$ to $\hat{p}_i$. Thus, this difference is also non-negative. 

Therefore, under any mixed strategy Nash equilibrium, a player $i$ can always modify one of their actions in some subgame to be deterministic and only increase their payoff. Therefore, they are still best responding. Furthermore, by Lemma~\ref{lem:purify-br}, all other non-responsive Agents remain best responding. 

Thus, we can iteratively turn a mixed strategy Nash equilibrium in the restricted game into a pure strategy Nash equilibrium. 
\end{proof}


\section{Remainder of Section \ref{sec:external}}

\subsection{Proof of Lemma \ref{lem:Agent-incentive}} \label{app:lem-agent-incentive}

\begin{proof}[Proof of Lemma \ref{lem:Agent-incentive}]
Suppose for contradiction that this is not the case.
Then, there is some strategy profile $p^*_{1 \ldots k}$ which is a non-responsive, pure strategy Nash equilibrium where there exists some Agent $i \in [k]$ such that their action $p^t_i(\pi^{1:t-1})$ given some transcript prefix $\pi^{1:t-1}$ is strictly less than the effort of a myopic optimal action $\tilde{p}^t_i$. Also let $\bar{y}^{1:t-1}$ be the state of nature sequence associated with this transcript prefix. Call $\bar{p}_{i}$ the Agent policy where Agent $i$ does exert the myopic optimal effort $\tilde{p}^t_i(\pi^{1:t-1})$ under the prefix $\pi^{1:t-1}$ and otherwise behaves exactly as $p^{*}_{i}$. As $p^{*}$ is a best response, we have that
\begin{align*}
     \mathcal{U}_{i}(1,\emptyset,p^{*}_{j \neq i},\tilde{p}_{i}) &- \mathcal{U}_{i}(1,\emptyset,p^{*}_{j \neq i},p^{*}_{i}) \geq 0 \\
    &  \rightarrow \mathbb{P}(g_{f,p^{*}_{j\neq i}, \bar{p}_{i}}(B^1_i) = \pi^{r,1:t-1}) \cdot \left(\mathcal{U}_{i}(t, \pi^{r,1:t-1}, p_{j \neq i},\tilde{p}_{i}) - \mathcal{U}_{i}(t, \pi^{r,1:t-1}, p^{*}_{j \neq i},p^{*}_{i}) \right) \geq 0 \\ 
    &  \rightarrow  \mathcal{U}_{i}(t, \pi^{r,1:t-1}, p_{j \neq i},\tilde{p}_{i}) - \mathcal{U}_{i}(t, \pi^{r,1:t-1}, p^{*}_{j \neq i},p^{*}_{i})  \geq 0,
\end{align*} 
where the first implication follows from Lemma \ref{lem:subgame-equivalence-general} and the second follows from the assumption that $B^1_i(\cdot)$ has non-zero support on the realized sequence of states of nature and that the transcript generating function (Definition \ref{def:transcript-generating-func}) generates a distribution over \emph{all} valid transcripts -- over the additional randomness of the Principal selections and the randomness used to sample from the selected Agent's outcome distribution in each round -- with the distribution over states of nature that it is given.

Thus, to show contradiction is is sufficient to prove that $ \mathcal{U}_{i}(t, \pi^{r,1:t-1}, p_{j \neq i},\tilde{p}_{i}) < \mathcal{U}_{i}(t, \pi^{r,1:t-1}, p^{*}_{j \neq i},p^{*}_{i})$.

Observe that Agent $i$ can weakly improve their expected immediate payoff at round $t$, by definition of a myopic optimal action.

Agent $i$ can also improve their continuation payoff for future rounds $t' > t$, denoted

\begin{align*}
   \mathbb{E}_{\pi^{t:T} \sim g_{f, p}(\bar{\pi}^{1:t-1},B^{t, t:T}_{i}(\bar{y}^{1:t-1}))} \left[\sum_{t' = t + 1}^{T} \mathbb{I}[f^t(\pi^{p, 1:t'-1}, \cR^t) = i] \cdot u_i(p^{t'}_{i}(\pi^{1:t'-1}), y^{t'}) \right], 
\end{align*}
by deviating to play $\tilde{p}^t_i$ in round $t$.
As all Agents are playing non-responsive strategies, the only effect Agent $i$'s action in round $t$ has on their continuation payoff is through the effect it has on the probability of their selection in future rounds. The probability of their future selection depends on the returns that Agent $i$ realizes, as the Principal selection mechanism only takes as input the history of realized returns. 
So, we can see that this is monotonically increasing in the Agent $i$'s effort in round $t$, due to the monotonicity of the Principal's selection mechanism (Definition \ref{def:monotone-alg}) and Assumption \ref{assumption:monotone-effort-return}, which establishes that increase in an effort level by Agent $i$ in a single round $t$ leads to higher expected returns. 

In particular, let the $a_1$ be deterministic action from strategy $p^{*, t}_i$ and $a_2$ be the higher effort level played deterministically from strategy $\tilde{p}^t_i$. Then, it is the case that
\begin{align*}
    \E_{o \sim \cD_{i, y, a_2}} [r(o)] \geq \E_{o \sim \cD_{i, y, a_1}} [r(o)].
\end{align*}

Thus, we know that Agent $i$'s deviation will only increase their continuation payoff via an increased selection probability in future rounds.

Therefore, we have shown that if Agent $i$ is playing a policy $p^*_i$ that in some round $t$ plays an effort level less than $\tilde{p}^t_i$ they can unilaterally deviate to strictly improve their payoff, and thus the strategy profile $p^*_{1 \ldots k}$ cannot be a non-responsive Nash equilibrium.
\end{proof}

\subsection{Proof of Lemma \ref{lem:myopic-policy}}
\begin{proof}[Proof of Lemma \ref{lem:myopic-policy}]
    Suppose for the sake of contradiction that there is some equilibrium strategy for agent $i$, denoted $p_i$, that is not a myopic optimal policy. If agent $i$ is playing $p_i$, then there is some round $t$ where they are not playing a myopic optimal action from their set $\tilde{A}^t_i$ for some transcript prefix $\hat{\pi}^{1:t-1}$. Let their payoff differ from the myopic payoff by $\epsilon$. We will consider some policy $\hat{p}_{i}$ such that $\hat{p}_{i}(\hat{\pi}^{1:t-1} = p_{i}(\hat{\pi}^{1:t-1})$, and for $\pi \neq \hat{\pi}^{1:t-1}$, $\hat{p}_{i}(\pi) = p_{i}(\pi)$. By Lemma~\ref{lem:subgame-equivalence-general}, the difference in the expected utility of the agent under $\hat{p}_{i}$ and $p_{i}$ is exactly 
    \begin{align*}
       \Pr[g_{f}(B^1_i) & = \hat{\pi}^{1:t-1}] \cdot \left(\mathcal{U}_{i}(t, \hat{\pi}^{1:t-1}, p_{j \neq i},\hat{p}_{i}) - \mathcal{U}_{i}(t, \hat{\pi}^{1:t-1}, p_{j \neq i},p_{i}) \right) \\
       & = \Pr[g_{f}(B^1_i) = \hat{\pi}^{1:t-1}] \cdot \big( \mathbb{E}_{\pi^{t} \sim g_f(\hat{\pi}^{1:t-1}, B^{t, t}_i(\bar{y}^{1:t-1}))}[u_{i}(\hat{p}_{i}(\hat{\pi}^{1:t-1}),y^t) - u_{i}(p_{i}(\hat{\pi}^{1:t-1}),y^t)  \\
       & \quad \quad + \mathbb{E}_{\pi^{t+1:T} \sim g_f(\hat{\pi}^{1:t-1} \cup \pi^t, B^{t:T, t}_i(\bar{y}^{1:t-1}) )}[\mathcal{U}_{i}(t+1, \hat{\pi}^{1:t-1}\cup \pi^{t}, p_{j \neq i},\hat{p}_{i}) -   \mathcal{U}_{i}(t+1, \hat{\pi}^{1:t-1}\cup \pi^{t}, p_{j \neq i},p_{i})]] \big) \\
        & = \Pr[g_{f}(B^1_i) = \hat{\pi}^{1:t-1}] \cdot \left( \mathbb{E}_{\pi^{t} \sim g(\hat{\pi}^{1:t-1}, B^{t, t}_i(\bar{y}^{1:t-1}))}[u_{i}(\hat{p}_{i}(\hat{\pi}^{1:t-1}),y^t) - u_{i}(p_{i}(\hat{\pi}^{1:t-1}),y^t) \right) \tag{By the fact that the action of the agent in the current round has no impact on their payoff in future rounds}\\
        & = \Pr[g_{f}(B^1_i) = \hat{\pi}^{1:t-1}] \cdot \epsilon \\
        & > 0 \tag{By the fact that the agent has full support beliefs.}
    \end{align*}

As agent $i$ can strictly improve their payoff by deviating to $\hat{p}_{i}$, $p_{i}$ is in fact not in equilibrium, leading to a contradiction.

\end{proof}

\section{Monotone No Swap-Regret Algorithms} \label{app:monomonomono}

\begin{algorithm}
\caption{The Blum-Mansour Swap Regret Algorithm \cite{blum2007external}} \label{alg-bm}
    \begin{algorithmic}
        \STATE {\bf Input} learning rate $\eta$
        \STATE Initialize $k$ copies of Algorithm \ref{alg-expw} with parameter $\eta$, where for each Agent $i \in [k]$ we maintain a distribution $q^t_i$
        \FOR{$t \in [T]$}
            \STATE Experience loss for each Agent, denoted $l_t = (l_t^1, \ldots, l_t^k)$
            \STATE Report loss vector $p_t^ii \cdot l_t$ to $i^\text{th}$ copy of Algorithm \ref{alg-expw}
            \STATE Update $q_t^i$ for each $i \in [k]$, each of the $k$ copies of Algorithm \ref{alg-expw}
            \STATE Combine $q_t^i$ for $i \in [k]$ into a single probability distribution over the $k$ Agents, by finding distribution $p_t$ such that $Ap_t = p_t$
        \ENDFOR
    \end{algorithmic}
\end{algorithm}

\begin{algorithm}
\caption{Exponential Weights Algorithm}
\label{alg-expw}
    \begin{algorithmic}
        \STATE {\bf Input} learning rate $\eta$
        \STATE For each Agent $i \in [k]$, initialize $w^1_i = 1$
        \STATE Define $w_t = \sum_{i \in [k]} w^t_it$
        \FOR{$t \in [T]$}
            \STATE Select Agent $i^t$ from distribution $p^t,$ where $p_t^i = w^t_i / w^t$ for all $i \in [k]$
            \STATE Observe loss $l^t$
            \STATE Update weight as $w^{t+1}_i = w^t_i(\exp{(-\eta \cdot l^t_i)})$ for all $i \in [k]$
        \ENDFOR
    \end{algorithmic}
\end{algorithm}

\subsection{Proof of Lemma~\ref{lem:blummansour-not-monotone}}\label{app:blum} 
\begin{proof}
Below we show an example with $T=100$ and the number of actions $k = 3$. Note that given a fixed loss sequence, Blum-Mansour's probability distributions it plays over each round are deterministic. Thus, we can numerically compute these probabilities to prove non-monotonicity.

Let the first $50$ loss vectors of $l_{1}$ each be $[-0.1,1,0]$. Let the final $50$ loss vectors of $l_{1}$ each be $[1,-1,0]$. Let $l_{2}$ be exactly the same as $l_{1}$, except that $l_{2}^1 = [-2,1,0]$. Then, when running Blum-Mansour as defined in Algorithm~\ref{alg-bm} with $\eta = 0.2$, the probability distribution maintained over actions at rounds $1:T$ (truncating the middle for clarity) is

\[
\begin{bmatrix}
0.34215564 & 0.31796216 & 0.33988219 \\
0.35085386 & 0.30294422 & 0.34620192 \\
0.35943921 & 0.28826134 & 0.35229945 \\
0.36791983 & 0.27390192 & 0.35817825 \\
0.37630081 & 0.25986054 & 0.36383865 \\
... & ... & ... \\
0.02854018 & 0.70527045 & 0.26618936 \\
0.02331378 & 0.73339873 & 0.24328749 \\
0.01873113 & 0.76082639 & 0.22044248 \\
0.01480289 & 0.78721832 & 0.19797879 \\
0.01151108 & 0.81226442 & 0.1762245\\
\end{bmatrix}
\]

While the corresponding distributions under $l_{2}$ are 
\[
\begin{bmatrix}
0.38028178 & 0.28913796 & 0.33058027 \\
0.38924322 & 0.27467575 & 0.33608103 \\
0.39809792 & 0.26052829 & 0.34137379 \\
0.40684786 & 0.24669652 & 0.34645563 \\
0.41549209 & 0.23318649 & 0.35132142 \\
... & ... & ... \\
0.02911013 & 0.72001336 & 0.25087651 \\
0.02359428 & 0.74841294 & 0.22799278 \\
0.0188028 & 0.77585784 & 0.20533936 \\
0.01473706 & 0.80200257 & 0.18326037 \\
0.01136643 & 0.82654706 & 0.1620865  \\
\end{bmatrix}
\]

Finally, we can look at the difference in move probabilities under $l_{1}$ and $l_{2}$:

\[
\begin{bmatrix}
2.90528236e-02 & -1.40423773e-02 & -1.50104462e-02\\
2.94279178e-02 & -1.38062655e-02 & -1.56216522e-02\\
2.98040066e-02 & -1.35855831e-02 & -1.62184235e-02\\
3.01780959e-02 & -1.33736323e-02 & -1.68044637e-02\\
3.05470484e-02 & -1.31640258e-02 & -1.73830227e-02\\
... & ... & ... \\
5.69949881e-04 & 1.47429088e-02 & -1.53128587e-02\\
2.80500303e-04 & 1.50142094e-02 & -1.52947097e-02\\
7.16695363e-05 & 1.50314536e-02 & -1.51031232e-02\\
-6.58236902e-05 & 1.47842499e-02 & -1.47184262e-02\\
-1.44649313e-04 & 1.42826443e-02 & -1.41379950e-02\\
\end{bmatrix}
\]
Note that at rounds $t=99$ and $t=100$, Blum-Mansour under $l_{1}$ plays action $1$ with a higher probability than Blum-Mansour under $l_{2}$. This violates the monotonicity condition.
\end{proof}

\subsection{Proof of Lemma \ref{lem:treeswap-monotone}} \label{app:treeswap-monotone}
\begin{proof}[Proof of \ref{lem:treeswap-monotone}]
    We begin with a general description of the TreeSwap algorithm. Refer to \cite{dagan2023external} for a more detailed description. 
    TreeSwap maintains multiple instances of a base no-external regret algorithm, which we refer to as \textsc{Alg}. Let $T$ be the total number of rounds that TreeSwap is run. Each instance of \textsc{Alg} is \emph{lazily} updated, meaning that each instance has some period length $m$, and is updated only $T/m$ times, during the end of each period, using the average loss of all the loss vectors given to the algorithm since the previous update. 
    In each round $t$, the action output by TreeSwap is the uniform mixture over the distributions output by a fixed number of these instances.
    The claim follows immediately by the fact that averaging preserves monotonicity. This means that neither the lazy updating of each instance of \textsc{Alg}, nor playing the averaging the distribution output by multiple instances of \textsc{Alg} in each round, violates the monotonicity of the base algorithm.
\end{proof}

\subsection{Monobandit Algorithm} \label{app:monobandit}

In addition to the definitions in Section~\ref{sec:online-prelims}, we introduce additional notation for a full-feedback learning algorithm, which is used in the implementation of our bandit-feedback algorithm MonoBandit.

\begin{definition}[Full-Feedback Learning algorithm $\A$, loss sequence $l$]
A full-feedback learning algorithm $\A_{f}$ consists of a deterministic function which takes as input $\cR^{\A_{f}}_{1:t}$, the realizations of randomness used by the algorithm in this round and all previous rounds, and $l^{1:t-1})$, the $k$-length loss vectors from all previous rounds. $\A_{f}$ outputs an action in the current round $t$.
\end{definition}

We can now define our algorithm MonoBandit implementing this reduction from full-feedback to bandit-feedback, which preserves the monotonicity and sublinear regret (either external or swap) guarantees of some full-feedback learning algorithm $\cA$. When we reason about the behavior of MonoBandit, we must be able to distinguish between the randomness inherent to $\cA$ and the randomness that the bandit algorithm uses in addition. The past realizations of both objects are visible to the adaptive adversary, but $\cA$ only has access to its own randomness. We therefore define the following distributions over randomness that are both utilized by the bandit-feedback algorithm:

\begin{definition}[$\cR$, $\cR_{f}$ and $\cR_{b}$]
Let $\cR$ be composed of the following two sets of randomness: $\cR_f$ and $\cR_b$. $\cR_f$ is the randomness utilized by the full-information algorithm $\cA$. $\cR_b$ is additional randomness utilized by MonoBandit, which is independent from $\cR_f$. In particular $\cR_b$ is composed of $T$ independent random variables $\cR^{t}_b$, where $\cR^t_{b} = EXPLOIT$ with probability $1 - \epsilon$ and $\cR^t_{b} = i$ with probability $\frac{\epsilon}{k}$ for all $i \in [k]$.
\end{definition}

\begin{algorithm} 
\caption{MonoBandit} \label{alg:monobandit}
    \begin{algorithmic}
        \STATE {\bf Inputs} $T$, Full-feedback Algorithm $\cA$
        \STATE {Parameter $\epsilon$}
        \STATE {\textbf{Probability vector $p$}       (initialized to the uniform distribution)}
        \STATE {\textbf{Sequence of sampled loss vectors $\hat{L}$} (initialized to empty)}
        \STATE {Randomness of full-feedback algorithm \textbf{$\cR_{f}$}}
        \STATE {Additional randomness of bandit algorithm \textbf{$\cR_{b}$}}
        \FOR{$t \in [T]$}
            \IF {$\cR_{b}^{t} \neq EXPLOIT$ }
                \STATE {Play action $\cR_{b}^{t}$}
                \STATE{Receive bandit feedback $b$}
                \STATE{$\hat{l}_{t}[\cR_{b}^{t}] \gets \frac{kb}{\epsilon}$}
                \STATE{$\hat{l}_{t}[i] \gets 0$ for all $i \neq \cR_{b}^{t}$}
                \STATE {$\hat{L} \gets \hat{L} \cup \hat{l}_{t}$}
                \ELSE 
                \STATE {Select action from probability distribution $p$}
                \STATE {$\hat{L} \gets \hat{L} \cup {0}^{k}$}
            \ENDIF
            \STATE {$p \gets \A(\hat{L}, \cR_{f}^{1:t})$}
        \ENDFOR
    \end{algorithmic}
\end{algorithm}

\subsection{Proof of Lemma \ref{thm:sr-reduction}} \label{app:sr-reduction}
\begin{proof}[Proof of Lemma \ref{thm:sr-reduction}]
Recall that we previously defined Swap Regret to be
\begin{align*}
    \textsc{SReg}_{f}(T, \ell) = \max_{\chi} \E_{\cR^{1:T}} \left[  \sum_{t=1}^T \ell(\phi^{1:t-1}_{\cR^{1:t-1}})[f(\phi^{1:t-1}_{\cR^{1:t-1}}, \cR^t) ] -   \ell(\phi^{t-1}_{\cR^{1:t-1}}[ \chi(f(\phi^{1:t-1}_{\cR^{1:t-1}}, \cR^t))] \right],
\end{align*}
where we use $\phi^{1:t-1}_{\cR^{1:t-1}}$ to denote the transcript prefix determinstically resulting from the algorithm $f$, adversary $\ell$ and a particular realization of $\cR^{1:t-1}$. Note that, given a particular $f$ and $\ell$, any $t$-length prefix of the transcript is uniquely defined by the realization of the randomness of the learning algorithm up until round $t-1$. Therefore, for simplicity for the entirety of this section we will write swap regret as
\begin{align*}
    \textsc{SReg}_{f}(T, \ell) = \max_{\chi} \E_{\cR^{1:T}} \left[  \sum_{t=1}^T \ell(\cR^{1:t-1})[f(\cR^{1:t}) ] -   \ell(\cR^{1:t-1})[\chi(f(\cR^{1:t})] \right]
\end{align*}

Let $\hat{l}_{t}$ represent the loss vector that is given to $\cA$ at round $t$. First, we will show that $\hat{l}_{t}$ is an unbiased sample of $\ell_{t}$, regardless of the behavior of the adaptive adversary. 
\begin{align*}
    \mathbb{E}_{\cR^{1:t}}[\hat{l}_{t}[i]] & = \mathbb{E}_{\cR^{1:t}}[\mathbb{I}[\cR^{t}_{b} = i] \cdot \frac{k}{\epsilon} \ell(\cR^{1:t-1})[i]] \\
    & = \frac{k}{\epsilon}\mathbb{E}_{\cR^{1:t}}[\mathbb{I}[\cR^{t}_{b} = i] \cdot  \ell(\cR^{1:t-1})[i]] \\
    & = \frac{k}{\epsilon}\mathbb{E}_{\cR^{t}_{b}}[\mathbb{I}[\cR^{t}_{b} = i]] \cdot  \mathbb{E}_{\cR^{1:t-1}}[\ell(\cR^{1:t-1})[i]] \tag{By the independence of $\cR^{t}_{b}$ and $\cR^{1:t-1}$} \\
    & = \mathbb{E}_{\cR^{1:t}}[\ell(\cR^{1:t-1})[i]]
\end{align*}

Let us write the swap regret of MonoBandit in terms of the swap regret incurred on EXPLORE and EXPLOIT rounds. For brevity in notation, we let $f$ refer to MonoBandit. 
\begin{align*}
\textsc{SReg}_{MonoBandit}(T, \ell) & \leq max_{\chi} \E_{\cR^{1:T}} \left[  \sum_{t \in EXPLORE} \ell(\cR^{1:t-1})[f(\cR^{1:t}) ] -   \ell(\cR^{1:t-1})[\chi(f(\cR^{1:t}))] \right] + \\ & max_{\chi} \E_{\cR^{1:T}} \left[  \sum_{t \in EXPLOIT} \ell(\cR^{1:t-1})[f(\cR^{1:t}) ] -   \ell(\cR^{1:t-1})[\chi(f(\cR^{1:t}))] \right]
\end{align*}

We will upper bound both terms, starting with the first. For all swap functions $\chi$, we have

\begin{align*}
& \E_{\cR^{1:T}} \left[  \sum_{t \in EXPLORE} \ell(\cR^{1:t-1})[f(\cR^{1:t}) ] -   \ell(\cR^{1:t-1})[\chi(f(\cR^{1:t}))] \right] \\
& \leq 
\E_{\cR^{1:T}} \left[  \sum_{t =1}^{T} \ell(\cR^{1:t-1})[f(\cR^{1:t}) ] -   \ell(\cR^{1:t-1})[\chi(f(\cR^{1:t}))] \right]
\\
& =
 \sum_{t =1}^{T} \E_{\cR^{1:t}} \left[ \ell(\cR^{1:t-1})[f(\cR^{1:t}) ] -   \ell(\cR^{1:t-1})[\chi(f(\cR^{1:t}))] \right]
\\
& =
 \sum_{t =1}^{T} \E_{\cR^{1:t}} \left[ \ell(\cR^{1:t-1})[\cA(\cR_{f}^{1:t}, \ell(\cR^{1:t-2})) ] -   \ell(\cR^{1:t-1})[\chi(\cA(\cR_{f}^{1:t},\ell(\cR^{1:t-2})))] \right] \tag{By the definition of MonoBandit}
\\
& =
 \sum_{t =1}^{T} \E_{\cR^{1:t}} \left[ \hat{l}[\cA(\cR_{f}^{1:t}, \ell(\cR^{1:t-2})) ] -   \hat{l}[\chi(\cA(\cR_{f}^{1:t},\ell(\cR^{1:t-2})))] \right] \tag{By the fact that $\hat{l}$ is an unbiased estimator of $\ell$ }
\\
& =
 \sum_{t =1}^{T} \E_{\cR^{1:t}} \left[ \hat{l}[\cA(\cR_{f}^{1:t}, \ell(\cR_{f}^{1:t-2})) ] -  \hat{l}[\chi(\cA(\cR_{f}^{1:t},\ell(\cR_{f}^{1:t-2})))] \right] \tag{By the fact that $\cR^{f}$ is independent of $\cR^{b}$ }
\\
& =
 \sum_{t =1}^{T} \E_{\cR_{f}^{1:t}} \left[ \hat{l}[\cA(\cR_{f}^{1:t}, \ell(\cR_{f}^{1:t-2})) ] -  \hat{l}[\chi(\cA(\cR_{f}^{1:t},\ell(\cR_{f}^{1:t-2})))] \right] 
\\
& \leq |\hat{l}_{max} - \hat{l}_{min}| \cdot R(T) \tag{By the definition of Swap Regret, and the assumption in the Lemma statement.} \\
& = \frac{k}{\epsilon} \cdot R(T) 
\end{align*}

As for the second term,

\begin{align*}
& max_{\chi} \E_{\cR^{1:T}} \left[  \sum_{t \in EXPLOIT} \ell(\cR^{1:t-1})[f(\cR^{1:t}) ] -   \ell(\cR^{1:t-1})[\chi(f(\cR^{1:t}))] \right] & \leq 
 \E_{\cR^{1:T}} \left[  \sum_{t \in EXPLOIT} 1 \right]
 \\ & = \epsilon T
\end{align*}

Putting these together, we can upper bound the swap regret of MonoBandit as

$$\frac{k}{\epsilon} \cdot R(T)  + \epsilon T $$

By setting $\epsilon = \sqrt{T\cdot k \cdot R(T)}$, we get that 

$$\textsc{SReg}_{MonoBandit}(T,\ell) \leq 2\sqrt{T \cdot k \cdot R(T)}, \ \forall \ell$$
\end{proof}

\subsection {Proof of Lemma \ref{thm:r-reduction}} \label{app:r-reduction}
\begin{proof}[Proof of Lemma \ref{thm:r-reduction}]
Recall that we previously defined external regret to be
\begin{align*}
    \textsc{Reg}_{f}(T, \ell) = \max_{i \in [k]} \E_{\cR^{1:T}} \left[  \sum_{t=1}^T \ell(\phi^{1:t-1}_{\cR^{1:t-1}})[f(\phi^{1:t-1}_{\cR^{1:t-1}}, \cR^t) ] -   \ell(\phi^{t-1}_{\cR^{1:t-1}}[i] \right],
\end{align*}
where we use $\phi^{1:t-1}_{\cR^{1:t-1}}$ to denote the transcript prefix determinstically resulting from the algorithm $f$, adversary $\ell$ and a particular realization of $\cR^{1:t-1}$. Note that, given a particular $f$ and $\ell$, any $t$-length prefix of the transcript is uniquely defined by the realization of the randomness of the learning algorithm up until round $t-1$. Therefore, for simplicity for the entirety of this section we will write external regret as
\begin{align*}
    \textsc{Reg}_{f}(T, \ell) = \max_{i \in [k]} \E_{\cR^{1:T}} \left[  \sum_{t=1}^T \ell(\cR^{1:t-1})[f(\cR^{1:t}) ] -   \ell(\cR^{1:t-1})[i] \right]
\end{align*}

Let $\hat{l}_{t}$ represent the loss vector that is given to $\cA$ at round $t$. First, we will show that $\hat{l}_{t}$ is an unbiased sample of $\ell_{t}$, regardless of the behavior of the adaptive adversary. 
\begin{align*}
    \mathbb{E}_{\cR^{1:t}}[\hat{l}_{t}[i]] & = \mathbb{E}_{\cR^{1:t}}[\mathbb{I}[\cR^{t}_{b} = i] \cdot \frac{k}{\epsilon} \ell(\cR^{1:t-1})[i]] \\
    & = \frac{k}{\epsilon}\mathbb{E}_{\cR^{1:t}}[\mathbb{I}[\cR^{t}_{b} = i] \cdot  \ell(\cR^{1:t-1})[i]] \\
    & = \frac{k}{\epsilon}\mathbb{E}_{\cR^{t}_{b}}[\mathbb{I}[\cR^{t}_{b} = i]] \cdot  \mathbb{E}_{\cR^{1:t-1}}[\ell(\cR^{1:t-1})[i]] \tag{By the independence of $\cR^{t}_{b}$ and $\cR^{1:t-1}$} \\
    & = \mathbb{E}_{\cR^{1:t}}[\ell(\cR^{1:t-1})[i]]
\end{align*}

Let us write the external regret of MonoBandit in terms of the external regret incurred on EXPLORE and EXPLOIT rounds. For brevity in notation, we let $f$ refer to MonoBandit. 
\begin{align*}
\textsc{Reg}_{MonoBandit}(T, \ell) & \leq max_{i \in [k]} \E_{\cR^{1:T}} \left[  \sum_{t \in EXPLORE} \ell(\cR^{1:t-1})[f(\cR^{1:t}) ] -   \ell(\cR^{1:t-1})[\chi(f(\cR^{1:t}))] \right] + \\ & max_{i \in [k]} \E_{\cR^{1:T}} \left[  \sum_{t \in EXPLOIT} \ell(\cR^{1:t-1})[f(\cR^{1:t}) ] -   \ell(\cR^{1:t-1})[i] \right]
\end{align*}

We will upper bound both terms, starting with the first. For all fixed actions $i$, we have

\begin{align*}
& \E_{\cR^{1:T}} \left[  \sum_{t \in EXPLORE} \ell(\cR^{1:t-1})[f(\cR^{1:t}) ] -   \ell(\cR^{1:t-1})[i] \right] \\
& \leq 
\E_{\cR^{1:T}} \left[  \sum_{t =1}^{T} \ell(\cR^{1:t-1})[f(\cR^{1:t}) ] -   \ell(\cR^{1:t-1})[i)] \right]
\\
& =
 \sum_{t =1}^{T} \E_{\cR^{1:t}} \left[ \ell(\cR^{1:t-1})[f(\cR^{1:t}) ] -   \ell(\cR^{1:t-1})[i] \right]
\\
& =
 \sum_{t =1}^{T} \E_{\cR^{1:t}} \left[ \ell(\cR^{1:t-1})[\cA(\cR_{f}^{1:t}, \ell(\cR^{1:t-2})) ] -   \ell(\cR^{1:t-1})[i] \right] \tag{By the definition of MonoBandit}
\\
& =
 \sum_{t =1}^{T} \E_{\cR^{1:t}} \left[ \hat{l}[\cA(\cR_{f}^{1:t}, \ell(\cR^{1:t-2})) ] -   \hat{l}[i] \right] \tag{By the fact that $\hat{l}$ is an unbiased estimator of $\ell$ }
\\
& =
 \sum_{t =1}^{T} \E_{\cR^{1:t}} \left[ \hat{l}[\cA(\cR_{f}^{1:t}, \ell(\cR_{f}^{1:t-2})) ] -  \hat{l}[i] \right] \tag{By the fact that $\cR^{f}$ is independent of $\cR^{b}$ }
\\
& =
 \sum_{t =1}^{T} \E_{\cR_{f}^{1:t}} \left[ \hat{l}[\cA(\cR_{f}^{1:t}, \ell(\cR_{f}^{1:t-2})) ] -  \hat{l}[i] \right] 
\\
& \leq |\hat{l}_{max} - \hat{l}_{min}| \cdot R(T) \tag{By the definition of External Regret, and the assumption in the lemma statement.} \\
& = \frac{k}{\epsilon} \cdot R(T) 
\end{align*}

As for the second term,

\begin{align*}
& max_{\chi} \E_{\cR^{1:T}} \left[  \sum_{t \in EXPLOIT} \ell(\cR^{1:t-1})[f(\cR^{1:t}) ] -   \ell(\cR^{1:t-1})[\chi(f(\cR^{1:t}))] \right] & \leq 
 \E_{\cR^{1:T}} \left[  \sum_{t \in EXPLOIT} 1 \right]
 \\ & = \epsilon T
\end{align*}

Putting these together, we can upper bound the external regret of MonoBandit as

$$\frac{k}{\epsilon} \cdot R(T)  + \epsilon T $$

By setting $\epsilon = \sqrt{T\cdot k \cdot R(T)}$, we get that 

$$\textsc{Reg}_{MonoBandit}(T,\ell) \leq 2\sqrt{T \cdot k \cdot R(T)}, \ \forall \ell$$

\end{proof}

\subsection{Proof of Lemma \ref{thm:mono-reduction}}
\begin{proof}[Proof of Lemma \ref{thm:mono-reduction}] \label{app:mono-reduction}
For ease of notation, throughout this proof we write the MonoBandit algorithm as taking as input the full feedback vector $l$, while keeping in mind that it can only change its state based on the bandit feedback it receives from sampling $l$ each round. Furthermore, we let $\hat{l}_{t}(l_{t})$ denote the sampled loss vector fed to the internal full-information algorithm $\cA$ by MonoBandit at round $t$. 

For any fixed sequence of losses $l$, the sampled loss vector $\hat{l}_{t'}$ that MonoBandit feeds to $\cA$ in round $t'$ is independent of the sampled loss vector $\hat{l}_{t}$ seen at round $t < t'$. To see this, note that  it depends only on the current loss vector at round $t'$ and the realization of $\cR^{b}_{t'}$, both of which are independent of any previous losses or realizations of randomness: 
\begin{align*}
\hat{l}_{t'} =
\begin{cases}
 \frac{k}{\epsilon} \cdot l_{t'}[i], \cR^{b}_{t'} = i\\
{0}^{k}, \cR^{b}_{t'} = EXPLOIT
\end{cases}
\end{align*}

Now, consider the impact that decreasing loss of a single action in one round has on the probability of success later in the game. Let the loss sequence $l^*$ be the same as the loss sequence $l$ except at round $\bar{t}$, action $a$ has decreased loss. Let the bandit feedback loss sequence $b^*$ be the loss sequence seen by MonoBandit under $l^{*}$, and let $b$ be the loss sequence seen by MonoBandit under $l$. Furthermore, let $v(l,a)$ represent a $k$-length vector which is zeroes everywhere but $l$ at index $a$.

On any exploit round $t > \bar{t}$: 
\begin{align*}
& \mathbb{P}_{\cR^{1:t}}(MonoBandit_{t}(l^{*,1:t-1},\cR^{1:t}) = a) \\ & = \mathbb{P}_{\cR^{1:t}}(MonoBandit_{t}(l^{*,1:t-1},\cR^{1:t}) = a|\cR_{b}^{t} \neq EXPLOIT)\cdot \mathbb{P}(\cR_{b}^{t} \neq EXPLOIT) \\ & + \mathbb{P}_{\cR^{1:t}}(MonoBandit_{t}(l^{*,1:t-1},\cR^{1:t}) = a|\cR_{b}^{t} = EXPLOIT)\cdot \mathbb{P}(\cR_{b}^{t} = EXPLOIT)\\
& = \frac{\epsilon}{k} + \mathbb{P}_{\cR^{1:t}}(MonoBandit_{t}(l^{*,1:t-1},\cR^{1:t}) = a|\cR^{t}_{b} = EXPLOIT)\cdot \mathbb{P}(\cR^{t}_{b} = EXPLOIT) \\ 
& = \frac{\epsilon}{k} + \mathbb{P}_{\cR^{1:t}}(\cA_{t}(\hat{l}^{1:t-1}(l^{*,1:t-1}),\cR^{f}_{1:t}) = a|\cR^{t}_{b} = EXPLOIT)\cdot (1 - \epsilon) \\
& = \frac{\epsilon}{k} + \mathbb{P}_{\cR^{1:t}}(\cA_{t}(\hat{l}^{1:t-1}(l^{1:t-1}),\cR^{f}_{1:t}) = a|\cR^{t}_{b} = EXPLOIT, \cR^{\bar{t}}_{b} \neq a)\cdot (1 - \epsilon)(1 - \frac{\epsilon}{k})\\ & + \mathbb{P}_{\cR^{1:t}}(\cA_{t}(\hat{l}^{1:t-1}(l^{*,1:t-1}),\cR^{f}_{1:t}) = a|\cR^{t}_{b} = EXPLOIT, \cR_{b}^{\bar{t}} = a)\cdot (1 - \epsilon)(\frac{\epsilon}{k}) \tag{By the fact that $\hat{l}$ is only different between $l^{*}$ and $l$ if action $a$ is sampled in round $\hat{t}$.} 
\end{align*}

We can compare this probability with that of the unaltered sequence $l$:

\begin{align*}
& \mathbb{P}_{\cR^{1:t}}(MonoBandit_{t}(b^{*,1:t-1},\cR^{1:t}) = a) - \mathbb{P}_{\cR^{1:t}}(MonoBandit_{t}(b^{1:t-1},\cR^{1:t}) = a) \\  
& = \mathbb{P}_{\cR^{1:t}}(\cA_{t}(\hat{l}^{1:t-1}(l^{*,1:t-1}),\cR^{f}_{1:t}) = a|\cR^{t}_{b} = EXPLOIT, \cR_{b}^{\bar{t}} = a)\cdot (1 - \epsilon)(\frac{\epsilon}{k}) - \\ & \mathbb{P}_{\cR^{1:t}}(\cA_{t}(\hat{l}^{1:t-1}(l^{1:t-1}),\cR^{f}_{1:t}) = a|\cR^{t}_{b} = EXPLOIT, \cR_{b}^{\bar{t}} = a)\cdot (1 - \epsilon)(\frac{\epsilon}{k}) \\
& = \mathbb{P}_{\cR^{1:t}}(\cA_{t}(\hat{l}^{1:t-1}(l^{*, 1:\bar{t}-1}) \cup \hat{l}_{\bar{t}}(l^{*}_{\bar{t}}) \cup \hat{l}^{\bar{t}+1:t-1}(l^{*, \bar{t}+1:t-1}),\cR^{f}_{1:t}) = a|\cR^{t}_{b} = EXPLOIT, \cR_{b}^{\bar{t}} = a)\cdot (1 - \epsilon)(\frac{\epsilon}{k}) - \\ & \mathbb{P}_{\cR^{1:t}}(\cA_{t}(\hat{l}^{1:t-1}(l^{1:\bar{t}-1}) \cup \hat{l}_{\bar{t}}(l_{\bar{t}}) \cup \hat{l}^{\bar{t}+1:t-1}(l^{\bar{t}+1:t-1})) = a|\cR^{t}_{b} = EXPLOIT, \cR_{b}^{\bar{t}} = a)\cdot (1 - \epsilon)(\frac{\epsilon}{k}) \\
& = \mathbb{P}_{\cR^{1:t}}(\cA_{t}(\hat{l}^{1:t-1}(l^{1:\bar{t}-1}) \cup v(\frac{k}{\epsilon}l^{*}_{\bar{t}}[a], a) \cup \hat{l}^{\bar{t}+1:t-1}(l^{\bar{t}+1:t-1}),\cR^{f}_{1:t}) = a|\cR^{t}_{b} = EXPLOIT, \cR_{b}^{\bar{t}} = a)\cdot (1 - \epsilon)(\frac{\epsilon}{k}) - \\ & \mathbb{P}_{\cR^{1:t}}(\cA_{t}(\hat{l}^{1:t-1}(l^{1:\bar{t}-1}) \cup v(\frac{k}{\epsilon}l_{\bar{t}}[a], a) \cup \hat{l}^{\bar{t}+1:t-1}(l^{\bar{t}+1:t-1})) = a|\cR^{t}_{b} = EXPLOIT, \cR_{b}^{\bar{t}} = a)\cdot (1 - \epsilon)(\frac{\epsilon}{k}) \\
& \geq 0 \tag{By the monotonicity of $\cA$}
\end{align*}

\end{proof}

\end{document}